\documentclass{LMCS}

\def\dOi{10(4:5)2014}
\lmcsheading%
{\dOi}
{1--29}
{}
{}
{Jun.~28, 2013}
{Dec.~\phantom04, 2014}
{}

\ACMCCS{[{\bf Software and its engineering}]: Software notations and
  tools---General programming languages---Language types---Concurrent
  programming languages; [{\bf Theory of computation}]: Semantics and
  reasoning---Program reasoning---Program verification}

\subjclass{D.1.3 Concurrent Programming,
  D.2.4 Software/Program Verification -- Formal methods.}

\usepackage{mathtools}
\usepackage{amssymb,amsmath}
\usepackage{xspace}
\usepackage{colortbl}
\usepackage{stmaryrd}
\usepackage{hyperref}

\usepackage{listings}
\usepackage{xspace}
\usepackage{courier}

\lstdefinelanguage{Actor}
{morekeywords={method, fields, actor, var, self, void, int, boolean, if, else, for, case, and, or, type},
sensitive=true,
morecomment=[l]{//},
morecomment=[s]{/*}{*/},
morestring=[b]",
literate={--}{{$--$}}2 {==}{{$==$}}2 {epsilon}{{$\epsilon$}}1 {neg}{{$\sim$}}1
}

\newcommand{\cal}{\mathcal}

\theoremstyle{plain}\newtheorem{proposition}[thm]{Proposition}
\theoremstyle{plain}\newtheorem{lemma}[thm]{Lemma}
\theoremstyle{plain}\newtheorem{theorem}[thm]{Theorem}
\theoremstyle{plain}

\newif\iftype \typefalse
\newif\ifconf \conffalse
\newif\ifcamera \camerafalse

\newcommand{\bigfract}[2]{\frac{^{\textstyle #1}}{_{\textstyle #2}}}
\newcommand{\rulename}[1]{{\sc(#1)}}
\newcommand{\rulenamex}[1]{\mbox{\scriptsize\sc(#1)}}

\def \mathrule #1#2#3{\begin{array}{l} 
                       {\rulenamex{#1}}
                       \\ \bigfract{#2}{#3}
                      \end{array}}

\def \mathax #1#2{\begin{array}{l} 
                   {\mbox{\scriptsize {{\sc (#1)}}} } 
                   \\ #2
                  \end{array}}

\newcommand{\subst}[2]{[\raisebox{.5ex}{\footnotesize$#1$}  /
                        \raisebox{-.5ex}{\footnotesize$#2$} ]}

\newcommand{\combinator}[2]{\begin{array}{c}{\small #1} \\ {\small #2} \end{array}}

\newcommand{\pinull}{{\tt 0}}
\newcommand{\invk}{\mbox{\tt !}}
\newcommand{\prefix}{\mbox{\tt .}}
\newcommand{\newact}[1]{{\tt new} \; #1}

\newcommand{\letin}[3]{{\tt let} \; #1 = #2 \; {\tt in} \; #3}

\newcommand{\State}{{\tt S}}
\newcommand{\StateT}{{\tt T}}
\newcommand{\lred}[1]{\stackrel{#1}{\longrightarrow}}
\newcommand{\lleadsto}[1]{\stackrel{#1}{\leadsto}}

\newcommand{\eqdef}{\stackrel{{\it def}}{=}}

\newcommand{\f}{{\tt f}}
\newcommand{\g}{{\tt g}}
\newcommand{\ite}{\mbox{{\tt ;}}}
\newcommand{\false}{{\it f\hspace{-3pt}f}}
\newcommand{\true}{{\it t\hspace{-2.3pt}t}}

\newcommand{\upd}{\mapsfrom~}

\newcommand{\var}[1]{{\it var}(#1)}
\newcommand{\fresh}[1]{{\it fresh}(#1)}
\newcommand{\free}[1]{{\it free}(#1)}
\newcommand{\fields}[1]{{\it fields}(#1)}
\newcommand{\pred}[1]{{\it Pred}(#1)}

\newcommand{\dom}[1]{{\it dom}(#1)}

\newcommand{\proj}[2]{{#1\!\!\downharpoonright_{#2}}}

\newcommand{\wt}[1]{\widetilde{#1}}
\newcommand{\adef}[1]{{\tt #1}}

\newcommand{\semantics}[1]{[\![ #1 ]\!]}

\newcommand{\xbar}{\wt{x}}
\newcommand{\eqdot}{\stackrel{\bullet}{=}}

\newcommand{\actor}{${\tt Actor}$}
\newcommand{\actro}{${\tt Actor^{\tt ro}}$}
\newcommand{\actba}{${\tt Actor_{\tt ba}}$}
\newcommand{\actroba}{${\tt Actor^{\tt ro}_{\tt ba}}$}
\newcommand{\actsl}{${\tt Actor^{\tt sl}}$}

\begin{document}

% \mainmatter              % start of the contributions
\title[Decidability Problems for Actor Systems]
      {Decidability Problems for Actor Systems\rsuper*}

\author[F.~S.~de Boer]{Frank S.~de Boer\rsuper a}
\address{{\lsuper a}CWI/LIACS,  The Netherlands}
\email{{\tt f.s.de.boer@cwi.nl}}
\thanks{{\lsuper{a-d}}This work has been has been supported by the HATS Project No.~FP7-231620 
(Highly Adaptable and Trustworthy Software using Formal Models) and by the
ENVISAGE Project No.~FP7-610582 (Engineering Virtualized Services) of the EC}

\author[M.~M.~Jaghoori]{Mohammad Mahdi Jaghoori\rsuper b}
\address{{\lsuper b}Leiden University, The Netherlands}
\email{{\tt jaghoori@cwi.nl}}
% \thanks{This work has been funded by the }

\author[C.~Laneve]{Cosimo Laneve\rsuper c}
\address{{\lsuper{c,d}}Department of Computer Science and Engineering, University of Bologna, INRIA Focus, Italy}
\email{\{cosimo.laneve,gianluigi.zavattaro\}@unibo.it}
% \thanks{thanks 1}

\author[G.~Zavattaro]{Gianluigi Zavattaro\rsuper d}
%\address{University of Bologna, INRIA Focus Research Team, Italy}
%\email{zavattar@cs.unibo.it}
% \thanks{thanks 1}

\keywords{Actors, RESTful services, decidability problems, 2-Counter
  Machines, well-structured transition systems, embedding relation.}

\titlecomment{{\lsuper*}This paper is a full version of an extended
  abstract that appears in~\cite{DJLZ2012}. With respect to the
  conference paper, this one contains examples and discussions, see
  Section~\ref{sec.examples}, and complete proofs of statements.}

\begin{abstract}
We introduce a nominal actor-based language and study its expressive power.
% which  besides dynamic creation of actors supports
%the dynamic creation of variable names which can be passed around  as parameters 
%in method calls.
We have identified the presence/absence of fields as a crucial feature:
the dynamic creation of  names
in combination with fields gives rise to Turing completeness.
On the other hand, restricting  to stateless actors 
gives rise to systems  for which properties such as termination are decidable.
This decidability result still holds for actors with states when
the number of actors is bounded and the state is read-only.
%in the presence of  actors with read-only fields (only) 
%we have further to restrict to a bounded number of actors.

% \bigskip
% 
% TO BE MERGED WITH:
%  We study an actor language that supports dynamic creation of statefull actors
%  and of values. We demonstrate that the full language is Turing complete, while
%   relevant fragments are not. In particular termination is decidable in
%  programs with stateless actors. If such programs also use finitely 
%  many actors then ...
%  \Cosimo{Gigio?}
%  %
%  We also argue that stateless actors correspond to stateless web services...

%\keywords{Actor model, asynchronous method invocation, Turing completeness, decidable problems.}
\end{abstract}

\maketitle              % typeset the title of the contribution

\ifcamera
\else
\pagestyle{plain}
\fi

\section{Introduction}

Since their introduction in~\cite{Hewitt69}, actor languages have  evolved as a powerful computational model for defining distributed and concurrent
systems~\cite{Agha90,Agha97}. 
Languages based on actors have been also designed for modelling embedded systems \cite{LeeLN09,LeeACtorEmbedded03}, 
wireless sensor networks \cite{CheongSensor05,RazaviBSKSS10}, multi-core programming \cite{KarmaniSA09}, and web services \cite{Chang-web-sac2007,Chang-DAIS07}.
The underlying concurrent model of actor languages also forms the basis of the   programming languages
Erlang \cite{Armstrong10Erlang} and Scala \cite{haller09tcs} that have recently gained in popularity, in part due to their support for scalable concurrency.

In actor languages~\cite{Agha90,Hewitt69,Sirjani06}, actors use a queue
for storing the invocations to their methods in a FIFO manner. The queued invocations
are processed sequentially by executing the corresponding method bodies.
The encapsulated memory of an actor is represented by a finite number of \emph{fields}
that can be read and set by its methods
and as such exist throughout  its life time.

In this paper we introduce a nominal actor-based language and study its expressive power. This language,  besides dynamic creation of actors, also supports
the dynamic creation of variable names that can be stored in fields and 
communicated in method calls.
As such our nominal actor-based language gives rise to unboundedness in
(1) internal queues of the actors, (2) dynamic actor creation/activation and (3) dynamic creation of variable names.

\emph{Statelessness} has recently been adopted as a basic principle 
of service oriented computing, in particular by RESTful services. 
Such services are designed to be stateless, and contextual information should be added to messages, so a service can customize replies simply by looking at the received request messages. 
In service oriented computing read-only fields (which are initialized upon activation)  are used to provide configuration/deployment information that distinguishes the distinct instances of the same service.
We have identified the presence/absence of fields as a crucial  feature
of our language:
 (1) and (3) in combination with fields gives rise to a Turing complete calculus.
On the other hand, restricting to stateless actors gives rise to systems for which properties such as termination 
\ifconf
\else
and process reachability 
\fi
are decidable. In order to preserve this decidability
result to actors with states we have to restrict the number of actors to be finite and the state to be read-only.

%web-services and many security and mobility protocols can be modelled by 
%systems composed of a bounded number of 
%actors with read-only fields only (which are initialized upon activation)
%but which may create and communicate an unbounded number of variable names.
More specifically,  we model a system  consisting of finitely many actors with 
read-only fields as a  well-structured transition system \cite{Finkel:2001} -- 
henceforth %the decidability of
\ifconf
termination.
\else
termination and process reachability are decidable.
% and  the inevitability problem.
\fi
Further, we show that an
abstraction of  systems of unboundedly many stateless actors 
(i.e., actors without fields) which preserves
termination and process reachability
is also an instance of well-structured transition system.
It turns out that, in the context of unbounded actor creation, 
this restriction to stateless actors is necessary by
a reduction to the halting problem for 2 Counter Machines.

To the best 
of our knowledge, the technique we use to establish the decidability results
for the above languages
is original since (\emph{i})  these systems respectively admit the creation of unboundedly many variables
 and actor names; (\emph{ii}) actors in general
are sensitive to the identity of names because of the presence of a name-match operator. 
In particular, in the case of finitely many actors with read-only fields,
we define an equivalence on process instances in terms of 
renamings of the variables that \emph{generate the same partition}.
This equivalence allows us to compute an upper bound to the
instances of method bodies, which is the basic argument for the model being a
well-structured transition system.
In case of systems %consisting of infinitely 
with unboundedly many stateless actors, 
the reasonable extensions of this equivalence on process
instances have been unsuccessful %(so far) 
because of the required abstraction of
the identity of actor names. Therefore we decided to apply our
arguments to an abstract operational model where messages
may be enqueued in every actor of the same class. The above equivalence can be  successfully used in this  model, thus yielding again the
upper bounds for the number of method body instances.  Further, the  abstract
model still provides enough information to derive
decidable properties of the language.

%\ifcamera
%{\it Related Works.}
%\else
\paragraph{\bf Related Works.}
%\fi
There exists a vast body of  related work on decidability of infinite-state systems 
(see  \cite{abdulla:96}).
However, to the best of our knowledge, %does not address 
the specific characteristics of 
the pure  asynchronous mechanism of queued and dequeued method calls in actor-based languages
has not been addressed.
%
%In \cite{Amadio2002} for example the authors study
%the decidability of the control reachability problem for various fragments of the asynchronous pi-calculus  focussing on
%the combination of name generation, name mobility, and unbounded control.
It is interesting to observe that %, in this context,
the most expressive known fragment of the  pi-calculus for which interesting verification problems are still decidable is the depth-bounded 
fragment~\cite{meyer08}. 
In~\cite{wies10} the theory of well-structured transition
systems is applied to prove the decidability of coverability problems for bounded depth pi-calculus.
Our nominal actor language also features the creation and communication of  new names. In our decidable fragments however, differently from the
depth-bounded pi-calculus fragment, 
we do not %assume 
restrict the creation and communication of names. For instance, in the queue of an actor we might have 
unboundedly many messages (representing process continuations) where each message shares one name with the previous message in the queue.
\ifcamera
\else
Decidability of asynchronous communication via a shared data space, 
read operation and non-blocking test operators on the shared space
of the coordination language Linda has been studied in \cite{Busi20009}.
In \cite{BouajjaniFQ07} the authors introduce an algorithm for reachability analysis of multithreaded object-oriented programs.
This algorithm requires both a bound on the number of context switches (between the threads) and the visible heap.
Further their approach is specific to object-oriented programs and is not applicable to  actor-based languages,
which are based on a different concurrency model.
As another final example, a system composed of finitely many actors  can be easily modeled
as a communicating finite state machine \cite{Brand83}, a formalism that is 
known to be Turing complete \cite{Brand83,Memmi85}.
However, this modelling does not scale to infinite actors and does not
display the execution model of actors.
\fi
%
%one loses the specific characteristics of the actor model of execution
%(and the dynamic creation of actors).
%Model-checking techniques and corresponding tool-support for the  state-based actor language Rebeca
%have been described in \cite{JaghooriMS06}.
Recent work on  actor-based language focusses on deadlock analysis:
In \cite{GiachinoL11}, a technique for the deadlock analysis has been introduced for a version of Featherweight Java
which features asynchronous method invocations and a synchronization mechanism based on futures variables.
The approach followed in \cite{deBoer2012} 
for detecting deadlock in an actor-like subset of Creol \cite{Johnsen07} is based on suitable over-approximations.

\section{The language {\actor}}
\label{sec.thelanguage}
Four disjoint infinite sets of names are used:
\emph{actor classes}, ranged over by $\adef{C}$, $\adef{D}$, $\cdots$,
%\emph{actor names}, ranged over $A$, $A'$, $B$, $B'$, $\cdots$,  
\emph{method names}, ranged over by $m$, $m'$, $n$, $n'$, $\cdots$, 
\emph{field names}, ranged over by ${\f}$, ${\g}$, $\cdots$, 
and \emph{variables}, ranged over by $x$, $y$, $z$, $\cdots$.
%We use $u$, $v$, $w$, $\cdots$ to range over field names and
%variables. 
For notational convenience, we use $\wt{x}$ when we refer to a list of variables $x_1,\dots,x_n$ (and similarly for other kinds of terms).
%; similarly we use $\wt{u}$ to range over a list of
%elements that may be either variables or fields.

The syntax of the language {\actor} uses \emph{expressions} $E$ and 
\emph{processes} $P$ defined by the rules
%{\small
\[
\begin{array}{rl}
E \quad ::=  &   \quad \f \quad |  \quad x \quad | \quad \newact{\adef{C}}(\wt{E})
\\
P \quad ::= & \quad \pinull  \quad | \quad  (\f \upd E) \prefix P 
 \quad | \quad \letin{x}{E}{P}
 \quad | \quad x \invk m(\wt{E}) \prefix P \quad |
 \\  & \quad
%\quad | \quad
  [E = E] P \ite P \quad 
%\\
  |  \quad  P+P 
\end{array}
\]
%}
%\Cosimo{changed $E \invk m(\wt{E}) \prefix P$ into $x \invk m(\wt{E}) \prefix P$}
%
An expression $E$ either  denotes 
a value stored in a field ${\f}$, or a variable $x$,
or a new actor of class $\adef{C}$ with fields initialized to the values of
$\wt{E}$.
%\sidenote{What if the conditional has no `else' statement? as in our example. I just used a new keyword skip.}%
%\sidenote{Is there any precedence between $\ite$ and $\prefix$? Should we use parentheses when needed? see Fig. \ref{fig:2CM_actor}}
%
A process may be either the terminated one, denoted by $\pinull$, or a field update 
$({\f} \upd E) \prefix P$, or the assignment $\letin{x}{E}{P}$ of a value to a variable, or an invocation 
$x \invk m(\wt{E}) \prefix P$ of a method $m$ of the actor $x$
with arguments $\wt{E}$, or a check $[E = E'] P \ite P'$ of the identity of 
expressions
with positive and negative continuations, or, finally  a nondeterministic process $P+P'$. 
%We use the term {\em method} definition as a synonym for the term process.
We never write the tailing $\pinull$ in processes; for example 
$(\f \upd x) \prefix \pinull$ will be always shortened into $(\f \upd x)$.
We will also shorten $[E = E'] P \ite \pinull$ into $[E = E'] P$.
Following the tradition of process calculi like CCS~\cite{Mil89},
we model sequences of actions by exploiting the 
so-called prefix notation (like, e.g., in $(\f \upd E) \prefix P $). 

The operation $\letin{x}{E}{P}$ is a binder of the occurrences of the
variable $x$ in the process $P$ that are not already bound by a nested
${\tt let}$ operation of $x$; the occurrences of $x$ in $E$ are \emph{free}.
Let $\free{P}$ be the set of variables of $P$ that are not bound. 
As usual, we identify processes $P$ and $P'$ that are equal up-to 
\emph{alpha-conversion} of bound names, written $P =_\alpha P'$.
The 
substitution operation $P \subst{y}{x}$ returns the process $P$ where the
free occurrences of $x$ are replaced by $y$. 
The 
substitution operation $P \subst{\wt{y}}{\wt{x}}$ returns the process $P$ where the
free occurrences of $\wt{x}$ are \emph{simultaneously} replaced by $\wt{y}$. In case,
alpha-renaming is required for avoiding name clashes. For example 
$\letin{z}{x}{\newact{\adef{C}}(y,z)}\subst{z,u}{x,y}$ returns
$\letin{z'}{z}{\newact{\adef{C}}(u,z')}$.

In the following examples and
encodings we shorten $\letin{x}{\f}{x \invk m(\wt{E}) \prefix P}$ into 
$\f\invk m(\wt{E}) \prefix P$ (we have preferred the simpler syntax to ease
the descriptions).

A \emph{program} is a \emph{main process} $P$ and a finite set of \emph{actor class 
definitions} $\adef{C} \prefix 
m(\xbar) = P_{\adef{C},m}$, where $P_{\adef{C},m}$ may contain the special 
variable ${\it this}$ (which can be seen as an implicit formal parameter of each method). In the following we restrict to programs that are
\begin{enumerate}
\item
\emph{unambiguous}, namely, every pair $\adef{C}$, $m$ has at most one definition;
\item
\emph{correct}, namely, let $\fields{\cdot}$ be
a map that associates a tuple of 
field names to every actor class. Then, 
(\emph{i}) in every expression $\newact{\adef{C}}(\wt{E})$, the length of the 
tuples $\wt{E}$ and $\fields{\adef{C}}$ are the same;
%for every method invocation $B \invk n(\wt{Z})$ in $P_{{\sf C},m}$ and in $P$, 
%then $B$ is an actor name of actor class {\sf D} and 
%there is a method definition ${\sf D} \prefix n(\wt{X}) = P_{{\sf D},n}$ such that 
%the length of $\wt{Z}$ and $\wt{X}$ is the same and have matching sorts; (i) 
(\emph{ii}) in every definition $\adef{C} \prefix 
m(\xbar) = P_{\adef{C},m}$, the field names 
occurring in $P_{\adef{C},m}$ are in the tuple  $\fields{\adef{C}}$.
%and no field name occurs in main process.
\end{enumerate}
In this paper, we abstract from types and type-correctness because we are
only interested in expressive power issues. However, it is 
straightforward to equip the above language with a type discipline.

\ifcamera
\else
\subsection{Examples}
\label{sec.examples}
To illustrate the features of the actor language we discuss four
examples.

\begin{exa}
The \emph{merger} service is an application that forwards 
to a back office server the values carried by two invocations of 
clients to methods ${\it first}$ and ${\it second}$. The \emph{merger}
freezes the forward as long as there is no invocation of either  
${\it first}$ or ${\it second}$ -- in the process calculi community,
the \emph{merger} is modelled by a join pattern~\cite{FG96}.
In our actor language the \emph{merger} is modelled by a class actor
$\adef{Merger}$ with six fields: ${\tt t}$ and ${\tt f}$ storing the  values \emph{true} and \emph{false}, respectively; ${\tt fst}$ and ${\tt snd}$
storing \emph{true} or \emph{false} according to the method 
${\it first}$ and ${\it second}$ have been invoked or not, respectively;
${\tt g}$ storing the  argument of the invocation; ${\tt srv}$ storing the
actor name of the back office server.

The class actor ${\adef{Merger}}$ 
includes the following three methods:
{\small
\[
\begin{array}{rll}
\adef{Merger} \prefix 
{\it init}(x) & = &
({\tt t} \upd \true) \prefix 
({\tt f} \upd \false) \prefix 
({\tt fst} \upd \false) \prefix
({\tt snd} \upd \false) \prefix
({\tt srv} \upd x)
\\
\\
\adef{Merger} \prefix 
{\it first}(a) & = &
[{\tt snd} = {\tt t}] \; {\tt srv} \invk {\it m}(a,{\tt g}) \prefix
 ({\tt snd} \upd {\tt f}) \ite 
 \\
 && \qquad \qquad 
\; [{\tt fst} = {\tt f}] \; ({\tt fst} \upd {\tt t}) \prefix 
 ({\tt g} \upd a) \ite {\it this} \invk {\it first}(a)
\\
\\
\adef{Merger} \prefix 
{\it second}(a) & = &
[{\tt fst} = {\tt t}] \; {\tt srv} \invk {\it m}({\tt g},a) \prefix
 ({\tt fst} \upd {\tt f}) \ite 
\\
&& \qquad \qquad 
\; [{\tt snd} = {\tt f}] \; ({\tt snd} \upd {\tt t}) \prefix 
 ({\tt g} \upd a) \ite {\it this} \invk {\it second}(a)
\end{array} 
\]
}%
The method ${\it init}$ manifests a basic feature of our actor language: 
the creation of new variables. In particular, $\true$ and $\false$ are free
variables in the method definition of ${\it init}$. When ${\it init}$ will be invoked, they 
will be replaced by two different fresh variables
(said in pi-calculus jargon~\cite{PIC}, we are 
assuming an implicit $(\nu \; \true,\false)$- at the beginning of the method body). These 
variables are stored in the fields ${\tt t}$ and ${\tt f}$, respectively and they 
will be used to 
update the fields ${\tt fst}$ and ${\tt snd}$ appropriately
.
The merger forwards two messages to the server, which is stored in the field 
${\tt srv}$. These messages have been received through invocations of {\it first} and {\it second}.
No forward occurs if one of the two methods has not been invoked. In particular,
{\it first} (the method {\it second} behaves in a similar way) performs the
forward if {\it second} has been evaluated \emph{and} its message has not
been already forwarded  (field ${\tt snd}$ equal to $\true$); in this case
the parameter of the invocation of {\it second} has been stored into ${\tt g}$. Otherwise, if there is no
previous invocation to {\it first}, still to be forwarded, (field ${\tt fst}$ equal 
to $\false$) then ${\tt fst}$ is set to $\true$ and the parameter stored into ${\tt g}$.
There is a possibility that {\it first} is evaluated and a previous evaluation 
of {\it first} has still to be forwarded (field ${\tt fst}$ equal 
to $\true$). In this case,  the invocation is bounced back (it is enqueued in the 
actor queue -- see the operational semantics).

%
%
%As a first example we consider three actors,
%a {\it client}, a {\it merger}, and a {\it server}.
%The {\it client}
%invokes the two methods {\it m1} and {\it m2}
%of the {\it merger}. As soon as both 
%of them are invoked, the {\it merger} invokes
%the method {\it m} of the {\it server}.
%We leave unspecified the classes ${\adef{C}}$
%and ${\adef{S}}$ of the {\it client} and {\it server},
%respectively; we simply assume that the former
%includes a method {\it init} and that the latter
%contains a method {\it m}.
%
%The main program that instantiates the {\it client},
%the {\it merger} and the {\it server} is as follows:
%$$
%\begin{array}{l}
%\letin{{\it client}}{\newact{\adef{C}}}{ \big(
%\letin{{\it merger}}{\newact{\adef{M}}}{ \big(
%\letin{{\it server}}{\newact{\adef{S}}}{\\
%{\it client} \invk {\it init}({\it merger}) \prefix 
%{\it merger} \invk {\it init}({\it server})
%}\big)}\big)}
%\end{array}
%$$
%
\end{exa}
\fi

\ifconf
\else
\begin{exa}
\newcommand{\username}{{\tt username}}
\newcommand{\password}{{\tt password}}
\newcommand{\auth}{{\tt auth}}
\newcommand{\ctoken}{{\tt ctoken}}
\newcommand{\atoken}{{\tt atoken}}
\newcommand{\server}{{\tt server}}
\newcommand{\sfield}{{\tt server}}
As a second example we model the
OpenID authentication protocol~\cite{openID}.
Three actors are considered: a {\it client}, a {\it server},
and an {\it  authProvider}.
The {\it client} sends a {\it login} request to the {\it server}.
The {\it server} generates two secure tokens {\it clientToken}
and {\it authToken}, used for secure communication with the 
{\it client} and the {\it  authProvider}, respectively,
and then sends the two tokens to the corresponding actors.
Subsequently, the {\it client} sends an {\em authentication}
message to the {\it  authProvider} that (after
checking the correctness of the username and password)
communicates to the {\it server} whether the authentication
{\it succeeds} or {\it fails}.
In the following, we abstract from the management of username
and password and we simply assume that they are stored in the
fields ${\tt u}$ and ${\tt p}$ of the {\it client} 
when it sends the {\it authenticate} message
to the {\it authServer}.

The class ${\adef{C}}$ of the {\it client}
includes the following two definitions:
$$
\begin{array}{rll}
\adef{C} \prefix 
{\it init}({\it server}) & = &
{\it server} \invk {\it login} ({\it this}) \\

\adef{C} \prefix 
{\it token}({\it authServer},{\it clientToken}) & = &
{\it authServer} \invk {\it authenticate} ({\tt u}, {\tt p},{\it clientToken})
\end{array} 
$$

The class ${\adef{S}}$ of the {\it server}
includes the following four definitions:
$$
\begin{array}{rll}
\adef{S} \prefix 
{\it init}({\it authServer}) & = &
(\auth \upd {\it authServer}) 
\\
\adef{S} \prefix 
{\it login}({\it client}) & = &
\auth \invk {\it open} ({\it this},{\it clientToken},{\it authToken}) \prefix \\
& &
{\it client} \invk {\it token} (\auth, {\it authToken})

\\
\adef{S} \prefix 
{\it succeeds}({\it authToken}) & = & \ldots
\\
\adef{S} \prefix 
{\it fails}({\it authToken}) & = & \ldots
\end{array} 
$$

The class ${\adef{A}}$ of the {\it authServer}
includes the following two definitions (we leave unspecified
the check of the correctness of username and password):
$$
\begin{array}{rll}
\adef{A} \prefix 
{\it open}({\it server},{\it cToken},{\it aToken}) & = &
(\server \upd {\it server}) \prefix (\ctoken \upd {\it cToken}) \prefix \\
& & (\atoken \upd {\it aToken})
\\
\adef{A} \prefix 
{\it authenticate}(u,p,{\it cToken}) & = & 
[\ctoken = cToken] \, 
\sfield \invk {\it succeeds}(\atoken) \ite \\
& &  \qquad \qquad \qquad \qquad  \ \sfield \invk {\it fails}(\atoken)
\end{array} 
$$

The main program that instantiates the {\it client},
the {\it server} and the {\it authProvider} is as follows:
$$
\begin{array}{l}
\letin{{\it client}}{\newact{\adef{C}}}{ }
\\
\qquad 
\letin{{\it server}}{\newact{\adef{S}}}{ }
\\
\qquad \qquad 
\letin{{\it authProvider}}{\newact{\adef{A}}}{}
\\
\qquad \qquad \qquad {\it client} \invk {\it init}({\it server}) \prefix 
{\it server} \invk {\it init}({\it authProvider})
%}\big)}\big)}
\end{array}
$$
\end{exa}

\begin{exa}
We illustrate the modeling of a register that stores natural numbers and supports the
operations of increment and decrement, when its value is positive.
In the absence of data-types in our nominal actor language, we model a register by an actor
and implement the value of the register by the
number of invocations of the method ${\it item}$ stored in its  queue.
Below we describe the corresponding class $R$, assuming a class $Ctrl$ which encodes the control of the register machine.
When an  operation
is performed, the register replies with an invocation ${\it run}({\it pc}, \true, \false)$, where ${\it pc}$ is a suitable continuation. 
The continuation is unique when the increment is invoked; the 
decrement has two possible continuations: a positive one, in case of success, and
a negative one (the register was 0, no decrement is performed).

\medskip

\; {\tt // $R$ has a field called} \ ${\tt dec}$, \quad {\tt ${\it Ctrl}$ is an actor
in the context}
\[\eqalign{
 R.{\it inc}({\it pc}, \true,\false) &=R \invk {\it item}(\true,
\false) \prefix {\it Ctrl} \invk {\it run}({\it pc}, \true,\false)
\cr
 R.{\it dec}({\it pc}, {\it pc}', \true,\false) &=({\tt dec} \upd \true)
 \prefix R \invk {\it checkzero}({\it pc}, {\it pc}', \true,\false)
\cr
 R.{\it checkzero}({\it pc}, {\it pc}', \true,\false) &=[{\tt dec} = \true]
 {\it Ctrl} \invk {\it run}({\it pc}', \true,\false) \prefix ({\tt dec}\upd \false)
\ite {\it Ctrl} \invk {\it run}({\it pc}, \true,\false)
\cr
 R.{\it init}(\true,\false) &=({\tt dec} \upd \false)
\cr
 R.{\it item}(\true,\false) &=[{\tt dec} = \false] R \invk {\it item}(\true,
\false) \ite ({\tt dec} \upd \false)
  }
\]
Method ${\it inc}$ simply gives rise to the storage of a new message  ${\it item}$ and 
a trigger of the continuation.
Executing a message ${\it item}$, amounts to bouncing the invocation back if the 
value of the field ${\tt dec}$ is false. In this case the invocation is enqueued again. 
Otherwise the invocations is ``consumed'' because there was a pending 
decrement to perform (see below) and ${\tt dec}$ is set to true (indicating that there is no pending decrement
anymore). Method ${\it dec}$ is the tricky one: it sets ${\tt dec}$ to true
and add an invocation of ${\it checkzero}$ (we are assuming that the register queue
always contains at most one invocation of ${\it dec}$). That is, we are postponing the
triggering of the continuation to the evaluation of  ${\it checkzero}$ that will
occur \emph{after} the evaluation of any other method in the queue of the register.
When ${\it checkzero}$ will be evaluated either (i) the ${\tt dec}$ field is true,
this means that the register did not contained any ${\it item}$ message -- its
value was 0 -- and the continuation ${\it pc}'$ is triggered; or (ii) 
the ${\tt dec}$ field is false, that is the register has been decremented and 
the  continuation ${\it pc}$ is triggered.

A refinement of the above register is used in Theorem~\ref{thm.undecidablestatefull} 
to demonstrate the Turing completeness of a sublanguage of {\actor}.
\end{exa}

\begin{exa}\label{ex:taskManager}
In Sections~\ref{sec.decidability} and~\ref{sec.stateless}
we will respectively prove decidability results for two fragments of our actor language:
in the first fragment only the main program can create new actors (so 
boundedly many actors can be created) and fields are read-only, while in the
second fragment actors have an empty state.
To illustrate these two fragments we show how they can be used to model
a system in which
$n$ workers are used to execute $n$ distinct tasks indicated by a client.

In the fragment with bounded actors and read-only fields the model can be described 
by considering a class $C$ for the client with an {\it init}
method responsible for invoking the task manager, passing him the description
of the $n$ tasks to be executed:

\; {\tt // $C$ has no fields} % called} \ ${\tt TManager}$,
% \quad {\tt ${\it Ctrl}$ is an actor
%in the context}
\[C.{\it init}({\it TManager}) =  {\it TManager} \invk {\it tasksExec}(t_1,\cdots,t_n)
\]

The task manager is an instance of a class providing a method called
{\it tasksExec} able to pass to the workers the description of their 
relative tasks:

\; {\tt // $TM$ has $n$ fields called} \ ${\tt W1},\cdots,{\tt Wn}$,
% \quad {\tt ${\it Ctrl}$ is an actor
%in the context}
\[TM.{\it tasksExec}(p_1,\cdots,p_n) =  {\tt W1} \invk {\it exec}(p_1)
\prefix \cdots \prefix {\tt Wn} \invk{\it exec}(p_n)
\]

Finally, a worker is an instance of a class $W$ with one method
{\it exec} able to execute the passed task (we leave this method
unspecified):

\; {\tt // $W$ has no fields and one method ${\it exec(p)}$} % (left unspecified)}

The main program first instantiates the {\it workers},
then the {\it task manager} by storing in its state the reference
to the workers, and finally the {\it client}; after the {\it init}
method of the client is invoked passing him a reference to the
task manager:
$$
\begin{array}{l}
\letin{{\it w_1}}{\newact{\adef{W}}}{}
\quad\cdots\quad 
\letin{{\it w_n}}{\newact{\adef{W}}}{ }
\\
\qquad \qquad 
\letin{{\it taskManager}}{\newact{\adef{TM}}s({\it w_1},\cdots,{\it w_n})}{}
\\
\qquad \qquad \qquad 
\letin{{\it client}}{\newact{\adef{C}}}{}
{\it client} \invk {\it init}({\it taskManager}) 
%\prefix 
%{\it server} \invk {\it init}({\it authProvider})
%}\big)}\big)}
\end{array}
$$

In the fragment with stateless actors we present an alternative specification
of the system in which the classes for the client and the workers are
as above, while the task manager is defined as follows:

\; {\tt // $TM$ has no fields} % called} \ ${\tt W1},\cdots,{\tt Wn}$,
% \quad {\tt ${\it Ctrl}$ is an actor
%in the context}
\[\eqalign{
 TM.{\it tasksExec}(p_1,\cdots,p_n) &= {\it this} \invk {\it nextTask}
(p_1,\cdots,p_n,x_1,x_2,\cdots,x_n,x_1)\cr
 TM.{\it nextTask}(p_1,\cdots,p_n,v_1,\cdots,v_{n+1}) &= 
[v_1 = v_2] \pinull \ite 
\letin{{\it w}}{\newact{\adef{W}}}{w \invk{\it exec}(p_1)}\prefix\cr
&\phantom{={\ }}
{\it this} \invk {\it nextTask}
(p_2,\cdots,p_n,p_1,v_2,\cdots,v_{n+1},v_1)\cr
  }
\]
Upon reception of the $n$ tasks to be executed, the 
task manager passes the task descriptions to the method
{\it nextTask}. This method instantiates one new
worker, passes to it the corresponding task, and then
recursively invokes itself to continue with
worker instantiations. In order to execute exactly $n$ recursive calls,
$n+1$ additional parameters are considered: in the first
call all such parameters are distinct excluding the first
and the last one, at every call the parameters are
circularly shifted, and when the first two parameters coincide
the chain of calls is terminated.
\end{exa}

\fi

\iftype
For simplicity in presentation, usually the method definitions and fields associated to an actor are collected together and referred to as the actor definition.
A program is then a collection of actor definitions.
We may use the syntax exemplified in Fig. \ref{fig:actor_def} to give the definition of an actor (on top of the syntax defined above for method definitions).
This figure defines an actor $A$, with fields $f$ and $g$, and methods $m$ and $m'$.
Method $m$ takes two variables $x$ and $y$ as its formal parameters, while $m'$ has no parameters.% 
%We use the keyword \ftt{boolean} to indicate that fields and variables can have the values \ftt{true} of \ftt{false}.
%A field or variable of type \ftt{int} can have any integer value%
\footnote{Assuming bounded or unbounded variables has not effect on the decidability results.}.%
\sidenote{Does having unbounded variables affect the decidable results?} 
$P_{A,m}$ and $P_{A,m'}$ represent the definitions of the methods.

\begin{figure}
\begin{lstlisting}[frame=single]
actor $A$ {
	fields { $f$, $g$ }
	method $m$($x$, $y$) = $P_{A,m}$
	method $m'$() = $P_{A,m'}$
}
\end{lstlisting}
\caption{An actor defined as the set of its fields and method definitions.}\label{fig:actor_def}
\end{figure}
\else

\fi

\ifconf
\paragraph{The operational semantics.}
\else
\subsection{The operational semantics}
\fi
The operational semantics of the language {\actor} will use an infinite set of 
\emph{actor names}, ranged over $A$, $B$, $\cdots$. This set
is partitioned by the actor classes in such a way that every partition retains 
infinitely many actor names. We write $A \in \adef{C}$ to say that $A$ 
belongs to the partition of $\adef{C}$.
In the following, the (\emph{run-time}) expressions 
will also include actor names and, with an abuse of notation, this extended
set of expressions
will be ranged over by $E$. The set of terms that are %either 
variables or 
actor names, called \emph{values}, will be addressed by $U$, $V$, $\cdots$.
Similarly for processes that, at run-time, may also have actor names. The extended
set will be addressed by $P$. We notice that $\free{P}$, when $P$ belongs to this
extended set, may returns variables \emph{and} actor names. We will also apply
$\free{\cdot}$ to tuples of (extended) expressions: $\free{  \wt{E}}$ returns
the set of variables and actor names in  $\wt{E}$.

The semantics is defined in terms of a \emph{transition relation}
$\State \lred{} \State'$, where $\State$, $\State'$, called \emph{configurations}, are
 sets of terms $A \triangleright (P, \varphi, q)$ with 
$A$ being an actor name,
$\varphi$, the \emph{state} of $A$, being a map from $\fields{\adef{C}}$ to values, 
where $A \in \adef{C}$,
and $q$ being a queue of terms $m(\wt{U})$. 
The empty queue will be denoted with $\varepsilon$.
Configurations contain at most one $A \triangleright (P, \varphi, q)$
for each actor name $A$. As usual, let $\lred{}^*$ be the transitive closure of 
$\lred{}$ and $\lred{}^+$ be $\lred{}\lred{}^*$.

%such that different terms $A \triangleright (P,q)$ and $B \triangleright (P',q')$ have
%$A \neq B$.
 
The operational semantics of {\actor} is defined in 
Table~\ref{tab.trans}, where the \emph{evaluation function} 
$E \lleadsto{\varphi} U \; ; \; \State$ is used (defined in Table~\ref{tab.eval}). This function 
takes an expression $E$ and a
store $\varphi$ and returns a value $U$ and a possibly empty configuration $\State$
of terms $A \triangleright (\pinull, \varphi, \varepsilon)$. These 
terms represent actors  created during the evaluation -- the names $A$ are \emph{fresh} --
and $\varphi$ records the initial values of the fields of $A$.
%
%. As for queues, the 
%empty sequence is denoted by $\varepsilon$ and we let $\wt{A} \cdot \varepsilon =
%\varepsilon \cdot \wt{A} = \wt{A}$. The function $\NEW{\wt{A}}$ returns a 
%sequence of terms $A \triangleright (\pinull, \varphi_\bot, \varepsilon)$, for
%every $A \in \wt{A}$, where 
%
The auxiliary function $\fresh{\cdot}$ used in the evaluation function
takes a class actor and returns an actor name of that class that is fresh.
In order to have a finitely branching transition system (see the remark
at the end of this section), we assume actor names 
in classes are totally ordered and $\fresh{\adef{C}}$ always
returns the first unused name of the class $\adef{C}$. In this way,
fresh names are selected in a deterministic manner.  
The same auxiliary function is used in rule \rulename{inst} 
 on a tuple of variables. In this case it returns a 
tuple of the same length of variables that are fresh. Also in this case, we
assume a total order of variables and $\fresh{\wt{x}}$, where $\wt{x}$ has length $n$, returns the least $n$ unused variables.  
%
%In Table~\ref{tab.opsem}, f
For notational convenience, we  always
omit the standard curly brackets in the set notation
and we use $``,"$ both to separate elements inside 
sequences and for set union (the actual meaning is made
clear by the context).
\begin{table} 
\[
\begin{array}{l}
U \lleadsto{\varphi} U \; ; \; \varnothing
\qquad \f \lleadsto{\varphi} \varphi(\f) \; ; \; \varnothing
\\
\\
\bigfract{ \wt{E} \lleadsto{\varphi} \wt{U}\; ; \; \State 
	\quad \wt{\f} = \fields{\adef{C}}
	\quad A = \fresh{\adef{C}}
	}{\newact{\adef{C}}(\wt{E}) \lleadsto{\varphi} A\; ; \; A \triangleright (\pinull, 
	[\wt{\f} \mapsto \wt{U}], \varepsilon), \State
	}
\\
\\
\bigfract{ E_i \lleadsto{\varphi} U_i \; ; \; \State_i, \quad \mbox{\rm for} \quad {i \in 1..n}
	}{
	E_1, \cdots , E_n \lleadsto{\varphi} U_1, \cdots , U_n \; ; \; \State_1, 
	\cdots , \State_n
	}
\\
\\
\end{array}
\]
\caption{\label{tab.eval} The evaluation relation $E \lleadsto{\varphi} U \; ; \; \State$}
\end{table}

\begin{table} 
%\begin{figure}[h]
\[
\begin{array}{l}
\mathrule{upd}{E \lleadsto{\varphi} U \; ; \; \State
	}{ 
	\begin{array}{l}
	A \triangleright ((\f \upd E) \prefix P, \varphi, q) 
	\lred{} A \triangleright (P, \varphi[\f \upd U], q), \State
	\end{array}
	} 
\\
\\
\mathrule{let}{E \lleadsto{\varphi} U \; ; \;\State
	}{ 
	\begin{array}{l}
	A \triangleright (\letin{x}{E}{P}, \varphi, q) \lred{} 
	 A \triangleright (P\subst{U}{x}, \varphi, q), \State
	 \end{array}
	} 
\\
\\
\mathrule{invk-s}{ \wt{E} \lleadsto{\varphi} \wt{U} \; ; \;\State
	}{
	\begin{array}{l}
 	A \triangleright (A \invk m(\wt{E}) \prefix P, \varphi, q)
         \lred{} 
	A \triangleright (P, \varphi,q \cdot  m(\wt{U})) , \State
	\end{array}
	}
\\
\\
\mathrule{invk}{ \wt{E} \lleadsto{\varphi} \wt{U} \; ; \; \State
	}{
	\begin{array}{l}
  	A \triangleright (A' \invk m(\wt{E}) \prefix P, \varphi, q) , 
	A' \triangleright (P', \varphi', q')
        \lred{} 
	A \triangleright (P,\varphi,q), A' \triangleright (P',\varphi', q' \cdot m
	(\wt{U})) , \State
	\end{array}
	}%\hfill (A \not = B) 
\\
\\
\mathrule{inst}{A \in \adef{C} \quad \adef{C} \prefix m(\wt{x}) = P 
	\quad 
	\wt{y} = \free{P} \setminus \wt{x}
	\quad
 	\wt{y'} = \fresh{\wt{y}}
	}{
	A \triangleright (\pinull, \varphi, m(\wt{U}) \cdot q)
	\; \lred{} \; 
	A \triangleright(P\subst{A}{{\it this}}\subst{\wt{y'}}{\wt{y}} 
	\subst{\wt{U}}{\wt{x}}, \varphi, q)
	}
\\
\\
\mathrule{match}{E , E' \lleadsto{\varphi} U,U \; ; \; \State
	}{
	A \triangleright ([E=E'] P \ite Q , \varphi, q)
	\; \lred{} \; A \triangleright (P , \varphi, q) , \State
	}
\\
\\
\mathrule{mmatch}{E,E' \lleadsto{\varphi} U,V \; ; \; \State
	  \quad U \neq V 
	}{
	A \triangleright ([E=E'] P \ite Q , \varphi, q)
	\; \lred{} \; A \triangleright (Q , \varphi, q) , \State
	}
\\
\\
\mathax{plus-l}{A \triangleright (P+Q, \varphi, q) 
	\; \lred{} \; A \triangleright (P, \varphi, q)
%	}{
%	A \triangleright(P + Q, q) , \State \; \lred{} \; 
%	\State'
	}
\\
\\
\mathax{plus-r}{A \triangleright (P+Q, \varphi, q) 
	\; \lred{} \; A \triangleright (Q, \varphi, q)
%A \triangleright (P, q) , \State 
%	\; \lred{} \; 
%	\State'
%	}{
%	A \triangleright(Q + P, q) , \State \; \lred{} \; 
%	\State'
	}
\\
\\
\mathrule{context}{\State 
	\; \lred{} \; 
	\State'
	}{
	\State , \State'' \; \lred{} \; 
	\State', \State''
	}
\\
% \bigfract{S \equiv S' ~~~ S' \to S'' ~~~ S'' \equiv S'''}
%       {S \to S'''} 
% \hfill \equiv \mbox{ is associativity and commutativity of ``,''}
\end{array}
\]
\caption{\label{tab.trans} The transition relation $\State \lred{} \State'$}
\end{table}

The initial configuration of a program with main process $P$  is
$\aleph \; \triangleright \; (P, \varnothing, \varepsilon)$, where $\aleph$ is a
name of the \emph{root}, an actor of a class
without fields and methods. We assume that the class of $\aleph$ does not
belong to the classes of the program. 
Note that the root actor is guaranteed to terminate because its queue remains empty
(no method invocation may be enqueued) and the main process (as any other
one) terminates.

We finally remark that transition systems of the language {\actor} are 
\emph{finitely branching} (every state has a finite number of successor states)
because the choices of fresh actor names (in the evaluation of 
$\newact{\adef{C}}$) and of fresh variables (in the instantiation of the bodies
of methods) are deterministic. 
%For example, if $\adef{C} \prefix m () = [x=x]P$ then $A 
% \triangleright (\pinull, \varnothing, m()) \lred{} 
% A \triangleright ([z=z]P, \varnothing, m())$ for every $z$. Additionally,
% every configuration  $A \triangleright ([z=z]P, \varnothing, m())$
% transits to $A \triangleright (P, \varnothing, m())$. Said otherwise, the sets
% $\Succ{\State} = \{ \State' \in {\cal S} \; | \; 
% \State \lred{} \State'\}$, called the \emph{successor configurations} of 
% $\State$, and $\pred{\State} = \{ \State' \in {\cal S} \; | \; 
% \State' \lred{} \State\}$, called the \emph{predecessor configurations} of
%  $\State$, are not finite, in general. 

%\ifcamera
%{\it Relevant sublanguages.}
%\else
\subsection{Relevant sublanguages}
%\fi
We will consider the following fragments of {\actor}
whose relevance has been already discussed in the Introduction:
\begin{itemize}[label=\actba]
\item[{\actba}] is the sublanguage where the ${\tt new}$ expression 
only occurs in the main process (the number of actor names that it is
possible to create is bounded).
\item[{\actro}] is the sublanguage without the field
update operation $(\f \upd E)$ (fields are read-only 
as they cannot be modified after the
initialization). 
\item[{\actroba}] is the intersection of {\actba} and {\actro}.
\item[{\actsl}] is the sublanguage with classes without fields
(objects are stateless).
\end{itemize}

\section{Undecidability results for {\actba} and {\actro}}
\label{sec.undecidability}
In this section we establish the main undecidability results 
%about the halting problem of
for the actor language in Section~\ref{sec.thelanguage}.
In particular, we will prove the undecidability of \emph{termination}
and \emph{process reachability}. 

\begin{defi}
\label{def.termandreach}
Consider an actor program.
The \emph{termination problem} is to decide whether it has no infinite computation;
the \emph{process reachability problem} is to decide,
given a process $P$, whether there exists  a computation of the program
traversing a configuration having a term $A \triangleright (P', \varphi, q)$
with $P'$ being equal
to $P$ up-to renaming of variables and actor names.
\end{defi}

% Actually, in order to convey a stronger result, we consider two sublanguages:
% (\emph{i}) where methods never use the ${\tt new}$ expression -- actors may be
% only created by the main process --, therefore the actor names are bounded, 
% and (\emph{ii}) where fields cannot be updated -- the fields are read-only after
% the initialization.

We will use a reduction technique 
%of the halting and reachability
%problems in 
from a Turing-complete model to %that of 
our actor model.
In particular, the Turing-complete models we consider are the 2 Counter Machines 
(2CMs)~\cite{Minsky67}
and a faulty variation of them.
A 2CM  is a machine with \emph{two registers} 
$R_1$ and $R_2$ holding arbitrary large natural numbers and a 
\emph{program} $P$ consisting of a finite sequence of numbered 
instructions of the following type:
\begin{itemize}
 \item $j: {\sf Inc}(R_i)$: increments $R_i$ 
 and goes to the instruction $j+1$;
 \item $j: {\sf DecJump}(R_i, l)$: if the content of $R_i$ is not 
 zero, then decreases it by 1 and goes to the instruction $j+1$, otherwise jumps to the instruction $l$;
 \item $j: {\sf Halt}$: stops the computation and returns the value in 
 the register $R_1$.
\end{itemize}
A state of the machine is given by a tuple $(i, v_1, v_2)$ where 
$i$ indicates the next instruction to execute (the program counter)
and $v_1$ and $v_2$ are the contents of the two registers. The
user has to provide the initial state of the machine. 
The transition relation of the 2CM will be denoted by $\Longmapsto$.

The faulty variation of the 2CM we use are the
``two Faulty Towers Machine'' (2FTM, for short)~\cite{FTM}. These machines have
\emph{two faulty registers} 
$R_1$ and $R_2$ holding either arbitrary large natural numbers or the \emph{faulty value} 
$\bot$. The program of a 2FTM is a finite sequence of numbered instructions that are 
the same of those in 2CMs. However, in contrast with 2CMs, 2FTMs from a state
$(i,v_1,v_2)$ may
\emph{nondeterministically} evolve into a faulty state $(i,\bot,v_2)$ or 
$(i,v_1,\bot)$ or
$(i,\bot,\bot)$. If an instruction $i$ is an {\sf Inc}/{\sf DecJump} that refers to
the register $R_1$ (respectively, $R_2$) then $(i,\bot,v)$ (respectively, 
$(i,v,\bot)$) evolves to $(0,\bot,v)$ (respectively, $(0,v,\bot)$).
In 2FTMs, the instruction numbered 0 is always assumed to be {\sf Halt}.
Let $\Longmapsto_{\tt F}$ be the transition relation of the  2FTMs.

By definition, every 2CM program with a 0-numbered instruction {\sf Halt} is 
a 2FTM program and conversely. If we restrict to 2CM with a 0-numbered instruction {\sf Halt},
it is easy to verify that every 
such machine has an infinite computation with a 2CM-semantics if and only if it
has an infinite computation with a 2FTM-semantics. Similarly for instruction reachability,
if we consider any non 0-numbered instruction.

In the sequel, we consider 2FTMs and 2CMs whose initial state is $(1,0,0)$.
%in which registers are initially set to zero.
% and where the instruction 
%0 is {\sf Halt}.

%Clearly in an actor \emph{programming language} with primitive infinite data types 
%like natural numbers, the modelling
%of 2CM is immediate. Without primitive data types, we have to device an 
%encoding of the natural numbers.

\subsection{The language {\actba}}
\label{ssec.finiteactors}
%
%The sublanguage we consider in this subsection is the one where the ${\tt new}$ expression only occurs in the main process -- the number of actor names that it is
% possible to create is finite.

%In this case, the basic idea is to 
We encode the value $n$ stored in a register as 
$n$ messages (of the same type) that are enqueued in an actor -- see 
Figure~\ref{fig:2CM_actor}. 
Namely, let $R_1$ and $R_2$ be two actors of class $\adef{R}$ and let 
the number of messages $\mathit{item}$ in $R_1$ and $R_2$ be 
their value. 
%This is the basic idea of the 2CM implementation
%illustrated in Figure~\ref{fig:2CM_actor}, where the control of the 2CM is managed by an  instance of class $\adef{Ctrl}$.
%
 \begin{figure}[t]
 \iftype
 \begin{lstlisting}[frame=single]
 actor $R_i$ {
 	fields { dec }
 	method one() =  dec $\upd$ $\false$ $\prefix$ ( [dec = $\false$] $R_i$ $\invk$ one())
 	method inc(pc) = $R_i$ $\invk$ one() $\prefix$ Ctrl $\invk$ run(pc)
 	method decJump(pc, pc') = 
 		dec $\upd$ $\true$ $\prefix$ $R_i$ $\invk$ checkZero(pc, pc')
 	method checkZero(pc, pc') =
 		[dec = $\true$] Ctrl $\invk$ run(pc') $\ite$ Ctrl $\invk$ run(pc)
 	method init() = dec $\upd$ $\false$
 }
 actor Ctrl {
 	method run(pc) =
 		[pc = 1] $\semantics{{\it Instruction_1}}$ $\ite$
 		. . . 
 		[pc = n] $\semantics{{\it Instruction_n}}$ $\ite$ $\pinull$
 	method init() = $R_1$ $\invk$ init() $\prefix$ $R_2$ $\invk$ init() $\prefix$ Ctrl $\invk$ run(1)
 }
 \end{lstlisting}
 \else
\[
\eqalign{%begin{array}{rrl}
 \llap{\hbox to 116pt{$\mathtt R$\hfill}}&
  \mbox{{\tt /\!/ $\mathtt{R}$ has fields $\mathtt{dec}$, $\mathtt{ctr}$,
    $\mathtt{loop}$ and $\mathtt{stop}$}}
 \cr
 {\tt R}.{\it item}(\true,\false) &= [{\tt stop} = \false]\big([{\tt dec} = \false] {\it this} \invk {\it item}(\true,\false) %\prefix ({\tt dec} \upd \false) 
 \ite ({\tt dec} \upd \false)\big)
 \cr
 {\tt R}.{\it inc}({\it pc}, \true,\false) &= 
 [{\tt stop} = \false]
 ({\tt loop} \upd \false) \prefix
 {\it this} \invk {\it item}(\true,
 \false) \prefix {\tt ctr} \invk {\it run}({\it pc}, \true,\false)
 \cr
 {\tt R}.{\it decjump}({\it pc}, {\it pc}', \true,\false) &= 
 [{\tt stop} = \false]
 ({\tt loop} \upd \false) \prefix
 ({\tt dec} \upd \true)
  \prefix {\it this} \invk {\it checkzero}({\it pc}, {\it pc}', \true,\false)
 \cr
 {\tt R}.{\it checkzero}({\it pc}, {\it pc}', \true,\false) &= 
 [{\tt stop} = \false]
 ({\tt loop} \upd \false) \prefix \cr
 &\phantom{={\ }} 
 \big( [{\tt dec} = \true] 
  {\tt ctr} \invk {\it run}({\it pc}', \true,\false) \prefix( {\tt dec}\upd \false)\ite {\tt ctr} \invk {\it run}({\it pc}, \true,\false) \big)
 \cr
 {\tt R}.{\it init}(\true,\false,{\it Ctrl}) &= ({\tt dec} \upd \false) 
 \prefix ({\tt ctr} \upd {\it Ctrl}) \prefix ({\tt loop} \upd \false)
 \prefix ({\tt stop} \upd \false)
 \prefix \cr
 &\phantom{={\ }} {\it this}\invk {\it bottom}(\true,\false)
 \cr
 {\tt R}.{\it bottom}(\true,\false) &= 
[{\tt loop} = \false] ({\tt loop} \upd \true).{\it this} \invk {\it bottom}(\true,\false); ({\tt stop} \upd \true)
  \cr\cr
 \llap{\hbox to 116 pt{$\mathtt{Ctrl}$\hfill}}&
   \mbox{{\tt /\!/ $\mathtt{Ctrl}$ has fields $\mathtt{stm}_1$, $\cdots$, 
   $\mathtt{stm}_n$ and $\mathtt{r}_1$ and $\mathtt{r}_2$}}
 \cr
 {\tt Ctrl}.{\it run}({\it pc}, \true,\false) &= [{\it pc} = {\tt stm}_1] 
 \semantics{{\it Instruction\_1}}_{1,\true,\false}
 %\semantics{{\it Instruction\_1}}_{1,\true,\false} \ite
 \cr
 &\phantom{={}} \quad \cdots
 \cr
 &\phantom{{}={}} [{\it pc} = {\tt stm}_n] \semantics{{\it Instruction\_n}}_{n,\true,\false}
 %\semantics{{\it Instruction\_n}}_{n,\true,\false}
 \cr
 {\tt Ctrl}.{\it init}() &= 
 %({\tt stm}_1 \upd x_1) \prefix \cdots \prefix 
 %({\tt stm}_n \upd x_n) \prefix ({\tt r}_1 \upd R_1) \prefix ({\tt r}_2 \upd R_2)
 %\\
 %&& 
 {\tt r}_1 \invk {\it init}(\true,\false,{\it this}) 
 \prefix {\tt r}_2 \invk {\it init}(\true,\false,{\it this}) \prefix {\it this} \invk 
 {\it run}({\tt stm}_1, \true,\false)
     \cr
  {\tt Ctrl}.{\it halt}() &= \pinull
  }
\]

\bigskip

\begin{itemize}
\item[] \hspace{-1.2cm}
\quad \quad where $\semantics{{\it Instruction\_i}}_{i,\true,\false}$ is equal to

\medskip
 \item[--]
 ${\tt r}_j \invk {\it inc}({\tt stm}_{i+1}, \true,\false)$\qquad if 
 ${\it Instruction\_i} = {\sf Inc}(R_j)$;
 \item[--]
 ${\tt r}_j \invk {\it decjump}({\tt stm}_{i+1}, {\tt stm}_{k}, \true,\false)$
 \qquad if 
 ${\it Instruction\_i} = {\sf DecJump}(R_j, k)$;
 \item[--]
 ${\it this}\invk {\it halt}$\qquad if 
 ${\it Instruction\_i} = {\sf Halt}$.
 \end{itemize}

\bigskip
 
\begin{itemize}
\item[]
\hspace{-1.2cm}
\quad\quad
The main process is 

\medskip

  $\letin{x}{\newact{\adef{Ctrl}}(x_1,\cdots,x_n,\newact{\adef{R}(\_,\_,\_,\_)},\newact{\adef{R}(\_,\_,\_,\_)})}{x \invk {\it init}()}$.
\end{itemize}
 \fi
 \caption{Encoding a 2FTM in {\actba} (``$\_$'' denotes an irrelevant initialization
 parameter)}\label{fig:2CM_actor}
 \end{figure}

The instruction {\sf Inc} is  implemented by inserting 
one  ${\it item}$  message in the queue of the corresponding register. 
In our formalism, this is done by invoking the method 
${\it item}$ whose execution has two possible outcomes: (i) 
the invocation is enqueued again; (ii) the invocation is discarded because
we are in the presence of a residual of a {\sf DecJump} operation, as described next.

In case (i), to avoid an infinite sequence of ${\it item}$ dequeues and 
enqueues, we introduce fields $\mathtt{stop}$ and $\mathtt{loop}$ which are  initialized to {\it false} and set to {\it true}
by the {\it bottom} method in case the queue contains only {\it item} messages.
This has as effect that the stored {\it item} messages are subsequently purged from the queue of the register. 
Note that differently from the example in Section~\ref{sec.thelanguage} we have encoded the boolean values
as fields (so that we do not need to pass them around).

% In case (i), to avoid an infinite sequence of ${\it item}$ dequeues and 
% enqueues, the queue of the registers is initialized 
% with a ${\it bottom}$ message.
% The execution of ${\it bottom}$ updates the field ${\tt loop}$ to $\true$
% (it is initialized to $\false$). This field is reset to $\false$
% when either ${\it inc}$, or ${\it decjump}$, or ${\it checkzero}$ is executed. 
% If the  ${\it bottom}$ method 
% is executed with ${\tt loop}$ set to $\true$, the register becomes inactive 
% by setting another field ${\tt stop}$. This value of ${\tt stop}$ possibly makes the 
% overall computation block as soon as an instruction concerning that register
% is performed.
%
% This ensures that the simulation
%of the 2CM progresses between two subsequent consumptions
%of the same ${\it item}$ message.

In case (ii), registers have a field ${\tt dec}$  that is set to ${\true}$ by 
a ${\it decjump}$ method execution. This field means that the actual decrement
of the register is delayed to the next execution of ${\it checkzero}$. Since in 
(ii) ${\it item}$ is not enqueued, then the register is actually decremented and
the field ${\tt dec}$  is set to ${\false}$. When 
${\it checkzero}$ will be executed, since 
${\tt dec} = \false$  then the next instruction of the 2CM is simulated.
On the contrary, when  ${\it checkzero}$ is executed with ${\tt dec} = \true$
then the decrement has not been performed (the register is 0) and the simulation 
jumps.

The {\sf Halt} instruction is simulated by invoking a method {\it halt} that does nothing.

\ifconf
Booleans
\else
As in the examples of Section~\ref{sec.thelanguage}, booleans
\fi
are implemented by two variables 
-- see the method $\adef{Ctrl}.{\it init}$ -- that are %passed around 
distributed during the invocations. 
% If this seems overwhelming to the 
% reader then it is possible to overcome this problem  by adding fields in the
% actors that retain the values true and false. We have preferred the current
% solution, but the other one is  easily implementable.
With a similar machinery, in the actor class $\adef{Ctrl}$,   the
labels of the instructions are represented 
by the variables $x_1,\cdots,x_n$, 
which are stored in the fields ${\tt stm}_1,\ldots,{\tt stm}_n$ of ${\it Ctrl}$.

\begin{thm}
\label{thm.undecidablestatefull}
%The 
Termination and process reachability 
%problems 
are undecidable in {\actba}.
\end{thm}
The undecidability of termination in {\actba} follows by the
property that a 2FTM  diverges if and only if the corresponding actor program
has an infinite computation. As regards process reachability, we need a smooth
refinement of the encoding in Figure~\ref{fig:2CM_actor} where the {\sf Halt}
instruction is simulated by a specific process. % $P'$ 
%(see Definition~\ref{def.termandreach}).

\ifconf
\else
\proof
Let us assume to have a fixed 2FTM and let $\semantics{{\it Instruction\_i}}_{i,\true,\false}$ be defined in Figure~\ref{fig:2CM_actor}. 
Let also
\[
\begin{array}{rll}
\semantics{(i,v_1,v_2)}_{\true,\false} \; \eqdef & 
\aleph \; \triangleright \; (\pinull, \varnothing, \varepsilon),  & (v_1 \neq \bot \; 
{\rm and} \; v_2 \neq \bot)
\\
& C \triangleright (\semantics{{\it Instruction\_i}}_{i,\true,\false}, \varphi_{{\tt ctrl}}, \varepsilon), 
\\
& R_1 \triangleright (\pinull, \varphi_{{\tt reg}}, {\it bottom}(\true,\false)
\cdot {\it item}(\true,\false)^{v_1} ),
\\
& R_2 \triangleright (\pinull, \varphi_{{\tt reg}}, {\it bottom}(\true,\false)
\cdot {\it item}(\true,\false)^{v_2})
\\
\\
\semantics{(i,\bot,v)}_{\true,\false} \; \eqdef & 
\aleph \; \triangleright \; (\pinull, \varnothing, \varepsilon),  & (v \neq \bot)
\\
& C \triangleright (\semantics{{\it Instruction\_i}}_{i,\true,\false}, \varphi_{{\tt ctrl}}, \varepsilon), 
\\
& R_1 \triangleright (\pinull, \varphi_{{\tt reg}}^{\bot}, \varepsilon ),
\\
& R_2 \triangleright (\pinull, \varphi_{{\tt reg}}, {\it bottom}(\true,\false)
\cdot {\it item}(\true,\false)^{v})
\\
\\
\semantics{(0,\bot,v)}_{\true,\false} \; \eqdef & 
\aleph \; \triangleright \; (\pinull, \varnothing, \varepsilon),  & (v \neq \bot)
\\
& C \triangleright (\pinull, \varphi_{{\tt ctrl}}, \varepsilon), 
\\
& R_1 \triangleright (\pinull, \varphi_{{\tt reg}}^{\bot}, \varepsilon ),
\\
& R_2 \triangleright (\pinull, \varphi_{{\tt reg}}, {\it bottom}(\true,\false)
\cdot {\it item}(\true,\false)^{v})
\end{array}
\]
where 
\begin{itemize}
\item[--] 
$\varphi_{{\tt ctrl}} \eqdef [{\tt stm}_i \mapsto x_i{}^{i \in 1..n}, {\tt r}_1 \mapsto R_1, {\tt r}_2 \mapsto R_2]$;

\item[--] 
$\varphi_{{\tt reg}} \eqdef [{\tt stop} \mapsto \false, {\tt loop} \mapsto \false,
{\tt dec} \mapsto \false, {\tt ctr} \mapsto C]$;

\item[--] 
$\varphi_{{\tt reg}}^\bot \eqdef [{\tt stop} \mapsto \true, {\tt loop} \mapsto \true,
{\tt dec} \mapsto \false, {\tt ctr} \mapsto C]$;

\item[--]
${\it item}(\true,\false)^{v} \eqdef \underbrace{{\it item}(\true,\false)
\cdots {\it item}(\true,\false)}_{v \; {\rm times}}$.
\end{itemize}
(the definitions of $\semantics{(i,v,\bot)}_{\true,\false}$ and 
$\semantics{(0,v,\bot)}_{\true,\false}$ and $\semantics{(0,\bot,\bot)}_{\true,\false}$
follow the same patterns).

We first observe that $\aleph \triangleright (P, \varnothing, \varepsilon) \lred{}^*
\semantics{(1,0,0)}_{\true,\false}$, where $P$ is the main process in 
Figure~\ref{fig:2CM_actor}.
Then we demonstrate the following properties:
\begin{enumerate}
\item[(1)]
if $(i,v_1,v_2) \Longmapsto_{\tt F} (j,v_1',v_2')$ then 
$\semantics{(i,v_1,v_2)}_{\true,\false} \lred{}^+ 
\semantics{(j,v_1',v_2')}_{\true,\false}$;

\item[(2)]
if $\semantics{(i,v_1,v_2)}_{\true,\false}$ has an infinite computation then 
the computation has infinitely many configurations like
\begin{eqnarray}
\label{eq.one}
\aleph \; \triangleright \; (\pinull, \varnothing, \varepsilon),   C \triangleright (\semantics{{\it Instruction\_i}}_{i,\true,\false}, \varphi_{{\tt ctrl}}, \varepsilon), 
R_1 \triangleright (P_1, \varphi_1, q_1),
R_2 \triangleright (P_2, \varphi_2, q_2)
\end{eqnarray}

\item[(3)]
if
\[
\begin{array}{@{\qquad}l}
\aleph \; \triangleright \; (\pinull, \varnothing, \varepsilon),   C \triangleright (\semantics{{\it Instruction\_i}}_{i,\true,\false}, \varphi_{{\tt ctrl}}, \varepsilon), 
R_1 \triangleright (P_1, \varphi_1, q_1),
R_2 \triangleright (P_2, \varphi_2, q_2)
\\
\lred{}^*
\\
\aleph \; \triangleright \; (\pinull, \varnothing, \varepsilon),   C \triangleright (\semantics{{\it Instruction\_j}}_{j,\true,\false}, \varphi_{{\tt ctrl}}, \varepsilon), 
R_1 \triangleright (P_1', \varphi_1', q_1'),
R_2 \triangleright (P_2', \varphi_2', q_2')
\end{array}
\]
is a computation with every intermediate configuration having the process of
the actor $C$ equal to $\pinull$ then there are two computations
\[
\begin{array}{@{\qquad}l}
\aleph \; \triangleright \; (\pinull, \varnothing, \varepsilon),   C \triangleright (\semantics{{\it Instruction\_i}}_{i,\true,\false}, \varphi_{{\tt ctrl}}, \varepsilon), 
R_1 \triangleright (P_1, \varphi_1, q_1),
R_2 \triangleright (P_2, \varphi_2, q_2)
\\
\lred{}^* \quad \semantics{(i,v_1,v_2)}_{\true,\false}
\end{array}
\]
and
\[
\begin{array}{@{\qquad}l}
\aleph \; \triangleright \; (\pinull, \varnothing, \varepsilon),   C \triangleright (\semantics{{\it Instruction\_j}}_{j,\true,\false}, \varphi_{{\tt ctrl}}, \varepsilon), 
R_1 \triangleright (P_1', \varphi_1', q_1'),
R_2 \triangleright (P_2', \varphi_2', q_2')
\\
\lred{}^* \quad \semantics{(j,v_1',v_2')}_{\true,\false}
\end{array}
\]
where the actor $C$ never moves such that
\[
(i,v_1,v_2) \Longmapsto_{\tt F} (j,v_1',v_2') \;.
\]
\end{enumerate}

\noindent The proof of (1) is a straightforward case analysis on the 
type of instruction $i$.

The proof of (2) uses an argument by contradiction. Assume that there are 
finitely many configurations like~\ref{eq.one}. Then, there is an infinite
suffix of this computation in which the actor $C$ does not perform actions.
This means that at least one actor $R_i$ performs infinitely many actions
by executing the methods {\it item} and {\it bottom}. But this is not possible
because {\it bottom} blocks the actor if it is invoked twice without
executing update actions on the register in between.

The proof of (3) simply considers two possible cases: either some
fields {\tt stop} is set to $\true$ during the computation or not.
In the first case the corresponding register of the 2FTM enters a faulty state
$\bot$. In the second case the instruction has been correctly
executed.

%\begin{eqnarray}
%\label{eq.one}
%(i,v_1,v_2) \Longmapsto (j,v_1',v_2') \quad \mbox{ if and only if } \quad
%\semantics{(i,v_1,v_2)}_{\true,\false} \lred{}^+ 
%\semantics{(j,v_1',v_2')}_{\true,\false}
%\end{eqnarray}

%It is possible to verify that if 
%the 2CM 
%executes the instructions $i_1,\cdots,i_n$ then 
%the corresponding actor program
%has a computation that executes the processes
%$\semantics{{\it Instruction\_{i_1}}}_{{i_1},\true,\false},\cdots,
%\semantics{{\it Instruction\_{i_n}}}_{{i_n},\true,\false}$ and conversely.
%Therefore the termination of the actor program implies
%the termination of the 2CM. The converse is not immediate because the actor
%program might have an infinite computation performing finitely many instructions
%$\semantics{{\it Instruction\_{i}}}_{{i},\true,\false}$. 
%This is excluded by
%the ${\it bottom}$ messages in registers' queues that block the entire computation
%in case a register consumes and reintroduces
%in its queue all its ${\it item}$
%messages without an actual progress of the 2CM
%simulation. Hence if the
%actor program has an infinite computation then
%the 2CM does not terminate.
%The above argument is further detailed in the Appendix.

The property (1) guarantees that if a 2FTM has an infinite computation
then also the corresponding encoding has an infinite computation.
The opposite follows from (1) and (2): if the encoding has an infinite
computation, it traverses infinitely many configurations representing
configurations of the corresponding 2FTM, thus it also has an infinite 
computation.
The undecidability result can be easily extended
to the process reachability problem. It suffices to
modify the process modeling an ${\sf Halt}$ (not numbered with 0)
by replacing ${\it this} \invk {\it halt}()$ with 
a process $Q$ different from all the other processes in
Figure~\ref{fig:2CM_actor}. We have that
$Q$ is reachable if and only if the 2CM with the same program of
the given 2FTM terminates.
\qed
\fi

%We conclude the section with few remarks.
%
%\Gigio{The remark 2 is in my opinion difficult to be 
%understood (at least, it is not clear to me)... is it necessary?}
%\begin{remark}
%\begin{enumerate}
%\item
%The reachability problem in the actor language is undecidable as well -- see the 
%proof of Theorem~\ref{thm.undecidablestatefull}.
%\item
%If we consider the subcalculus where the operations $[E = E'] \; P \ite Q$ are
%such that $Q = \pinull$ then the undecidability result of
%Theorem~\ref{thm.undecidablestatefull} may still be demonstrated (in some 
%nondeterministic form).
%\item
%The encoding in Figure~\ref{fig:2CM_actor} uses a  finite number of variables.
%\end{enumerate}
%\end{remark}

\subsection{The language {\actro}}
\label{ssec.readonly}
We show that {\actro}
is Turing-complete by means of an encoding of a 2CM
-- see Figure~\ref{fig:2CM_actor2}. 
In this encoding the two registers
are represented by two disjoint stacks of actors linked by the $\mathtt{next}$ field.
The top elements of the two stacks are passed as parameters $r_1$ and $r_2$ 
of the {\it run} method
of the controller. As before, this actor encodes the control of the 2CM.

The instruction {\sf Inc} is implemented by pushing an element on top of the
corresponding stack. This element is an actor of class $\mathtt{R}$ 
storing in its field 
the old pointer of the stack. The new pointer, \emph{i.e.} the new actor name,
is
passed to the next invocation of the {\it run} method.

The instruction {\sf DecJump} is implemented by popping the corresponding
stack. In particular, the method ${\it run}$ of the controller is invoked with the
field $\mathtt{next}$ of the register being decreased. This pop operation is
performed provided the register that is argument of ${\it run}$ is different
from ${\it nil}$. Otherwise a jump is performed.
Note that the other top of the stack $r_j$ ($i\not=j$) and the next instruction
to be executed are simply passed around and therefore they do not need
to be stored in updatable fields.

\begin{thm}
\label{thm.undecidablestateless}
%The 
Termination and process reachability 
%problems 
are undecidable in {\actro}.
\end{thm}

\ifconf
\else

\proof
It is easy to verify that if a 2CM has a computation
\[
(1,0,0) \Longmapsto (i_1,v_1,v_1') \Longmapsto \cdots \Longmapsto (i_n,v_n,v_n')
\]
then there is a computation
\[
\begin{array}{rl}
\aleph \triangleright (P, \emptyset, \varepsilon) \lred{}^* &
\aleph \triangleright (\pinull, \emptyset, \varepsilon), 
C \triangleright (\semantics{{\it Instruction_1}}_{{\it nil},{\it nil}}, \varphi, \varepsilon)
\\
\lred{}^* & 
\aleph \triangleright (\pinull, \emptyset, \varepsilon), 
C \triangleright (\semantics{{\it Instruction\_{i_1}}}_{r_1,r_1'}, \varphi, \varepsilon),
{\cal R}_{r_1}, {\cal R}_{r_1'}, {\cal G}_1
\\
\lred{}^* & 
\aleph \triangleright (\pinull, \emptyset, \varepsilon), 
C \triangleright (\semantics{{\it Instruction\_{i_n}}}_{r_n,r_n'}, \varphi, \varepsilon),
{\cal R}_{r_n}, {\cal R}_{r_n'}, {\cal G}_n
\end{array}
\]
where $P$ is the main process of Figure~\ref{fig:2CM_actor2},
$\semantics{{\it Instruction\_{i_j}}}_{r_j,r_j'}$ are defined in 
Figure~\ref{fig:2CM_actor2},
$\varphi = [ \mathtt{stm}_1 \mapsto x_1, \cdots, \mathtt{stm}_n \mapsto x_n,
\mathtt{nil} \mapsto {\it nil}]$, ${\cal R}_{r_i}$ and ${\cal R}_{r_i'}$ are
stacks of register actors whose length is $v_i$ and $v_i'$, respectively. For
instance, ${\cal R}_{r_1}$ of length $k$ is
\[
r_1 \triangleright (\pinull, [\mathtt{next} \mapsto r_2], \varepsilon),
r_2 \triangleright (\pinull, [\mathtt{next} \mapsto r_3], \varepsilon),
\cdots ,
r_k \triangleright (\pinull, [\mathtt{next} \mapsto {\it nil}], \varepsilon) \; .
\]
The configurations ${\cal G}_i$ only contain register  terms $r \triangleright (\pinull,
[\mathtt{next} \mapsto r'], \varepsilon)$ and represent \emph{garbage} (they are
inactive).

In contrast with Theorem~\ref{thm.undecidablestatefull}, the converse 
implication (every computation of the  {\actro} program in Figure~\ref{fig:2CM_actor2}
may be split in subcomputations of finite lengths that correspond to 
2CM transitions) is not difficult because the program of Figure~\ref{fig:2CM_actor2}
is deterministic.

The above correspondence guarantees that the 
computation of the actor system
terminates if and only if 
a ${\sf Halt}$ instruction is reached.
The undecidability of process reachability
is proved by using the same arguments
of Theorem~\ref{thm.undecidablestatefull}.
\qed
\fi

\begin{figure}[t]
\iftype
\begin{lstlisting}[frame=single]
actor $R_i$ {
	fields { dec }
	method one() = ( [dec = $\false$] $R_i$ $\invk$ one() $\ite$ skip ) $\prefix$ dec $\upd$ $\false$
	method inc(pc) = $R_i$ $\invk$ one() $\prefix$ Ctrl $\invk$ run(pc)
	method decJump(pc, pc') = 
		dec $\upd$ $\true$ $\prefix$ $R_i$ $\invk$ checkZero(pc, pc')
	method checkZero(pc, pc') =
		[dec = $\true$] Ctrl $\invk$ run(pc') $\ite$ Ctrl $\invk$ run(pc)
	method init() = dec $\upd$ $\false$
}
actor Ctrl {
	method run(pc) =
		[pc = 1] ${\it Instruction_1}}$ $\ite$
		. . . 
		[pc = n] ${\it Instruction_n}}$ $\ite$ $\pinull$
	method init() = $R_1$ $\invk$ init() $\prefix$ $R_2$ $\invk$ init() $\prefix$ Ctrl $\invk$ run(1)
}
\end{lstlisting}
\else
\[
\eqalign{
 \llap{\hbox to 116 pt{$\mathtt R$\hfill}}&
 \mbox{{\tt /\!/ $\mathtt{R}$ has a field $\mathtt{next}$}}\cr
{\tt R}.{\it dec_1}({\it ctrl},r,{\it stm}) &=
{\it ctrl} \invk {\it run}(\mathtt{next},r,{\it stm}) \cr
{\tt R}.{\it dec_2}({\it ctrl},r,{\it stm}) &=
{\it ctrl} \invk {\it run}(r,\mathtt{next},{\it stm}) \cr\cr
 \llap{\hbox to 116 pt{$\mathtt{Ctrl}$\hfill}}&
  \mbox{{\tt /\!/ $\mathtt{Ctrl}$ has fields 
  $\mathtt{stm}_1$, $\cdots$, 
  $\mathtt{stm}_n$ and $\mathtt{nil}$}}
\cr
{\tt Ctrl}.{\it run}(r_1,r_2,{\it pc}) &= [{\it pc} = {\tt stm}_1] 
\semantics{{\it Instruction\_1}}_{r_1,r_2} \ite 
\cr
& \qquad \cdots 
\cr
&\phantom{{}={}} [{\it pc} = {\tt stm}_n] 
\semantics{{\it Instruction\_n}}_{r_1,r_2} 
  }
\]

\bigskip

\leftline{where $\semantics{{\it Instruction\_i}}_{r_1,r_2}$ is equal to}
\begin{itemize}[label=--]
\item[--]
%$\letin{x}{\newact{\adef{R}}({r}_1)}{
${\it this} \invk {\it run}({\newact{\adef{R}}({r}_1)},{r}_2,\mathtt{stm}_{i+1})$
\qquad if 
${\it Instruction\_i} = {\sf Inc}(R_1)$;

\medskip

\item[--]
${\it this} \invk {\it run}(r_1,{\newact{\adef{R}}({r}_2)},\mathtt{stm}_{i+1})$
\qquad if 
${\it Instruction\_i} = {\sf Inc}(R_2)$;

\medskip

\item[--]
$[{r}_1 = \mathtt{nil}] {\it this} \invk {\it run}(r_1,r_2,\mathtt{stm}_k) \ite
{r}_1 \invk {\it dec_1}({\it this},{r}_2,\mathtt{stm}_{i+1})$ \\
\hfill
if ${\it Instruction\_i} =$ ${\sf DecJump}(R_1, k)$;

\medskip

\item[--]
$[{r}_2 = \mathtt{nil}] {\it this} \invk {\it run}(r_1,r_2,\mathtt{stm}_k) \ite
{r}_2 \invk {\it dec_2}({\it this},{r}_1,\mathtt{stm}_{i+1})$\\
\hfill
if ${\it Instruction\_i} = $ ${\sf DecJump}(R_2, k)$;

\medskip

\item[--]
$\pinull$ \qquad if 
${\it Instruction\_i} = {\sf Halt}$.
\end{itemize}\medskip

\leftline{The main process is 
$\letin{x}{\newact{\adef{Ctrl}}(x_1,\cdots,x_n,{\it nil})}{x \invk {\it run}({\it nil},{\it nil},x_1)}$.}
\fi
\caption{Encoding a 2CM in {\actro}}\label{fig:2CM_actor2}
\end{figure}

\section{Decidability results for {\actroba}}
%\section{Decidability results for programs with a bounded number of 
%actors with read-only fields}
%Actors Systems and WSTS}
\label{sec.decidability}

%Given Theorem~\ref{thm.undecidablestatefull}, in this section we study a 
%sublanguage of the one considered in Section~\ref{ssec.finiteactors}
%(methods never use the $\newact{}$ operation), namely  
%methods never use the $\newact{}$ operation \emph{and}
%the update operation $(\f \upd E)$. That is, in addition to the constraint
%of Section~\ref{ssec.finiteactors}, here  fields are
%read-only: they can be initialized and cannot be updated.
% --
%such actors are called  \emph{stateless} --
%and (\emph{ii}) methods never use the $\newact{}$ operation -- 
%$\mathtt{new}$ may be only performed by the main process. This means that
%actors can be only created  by the main process, which in turn means that there are
%finitely many of them.
%We notice that, 
%due to (i)
%in stateless
%actors, processes have no field update operation and processes 
%$[E = E']~P \ite Q$ are such that $E$ and $E'$ are either variables or 
%actor names.
%The rules of the operational semantics of stateless actors
%are simply obtained from Table~\ref{tab.opsem} by removing the store $\varphi$
%and the rule for field update.

We demonstrate that programs in {\actroba}
%consisting of finitely many 
%actors with read-only fields 
%can be modelled as
are well-structured transition systems~\cite{abdulla:96,Finkel:2001}.  This
will allow us to decide a number of properties, such as termination.
%entail our decidability results about this fragment of the actor language.
%
We begin with some background on well-structured transition systems.

A reflective and transitive relation is called \emph{quasi-ordering}.
A \emph{well-quasi-ordering} is a quasi-ordering $(X, \le)$ such that, for every infinite sequence $x_1, x_2, x_3,$ 
$\cdots$, there exist $i,j$ with $i < j$ and $x_i \leq x_j$.

\begin{defi}%[WSTS]
A \emph{well-structured transition system} %(WSTS) 
is a finitely branching transition system $({\cal S}, \lred{}, \preceq)$ where $\preceq$ is a quasi-ordering relation on states such that 
\begin{enumerate}
\item 
$\preceq$ is a well-quasi-ordering
\item 
$\preceq$ is upward compatible with $\lred{}$, i.e., for every $\State_1 
\preceq \State_1'$ and  $\State_1 \lred{} \State_2$, there exists $\State_1' 
\lred{}^* \State_2'$ such that $\State_2 \preceq \State_2'$.
\end{enumerate}
\end{defi}

\noindent Given a state $s$ of a well-structured transition system, $\pred{s}$
denotes the set of immediate predecessors of $s$ (i.e., $\pred{s}=\{s'\; |\; s' \lred{} s\}$)
while $\uparrow s$ denotes the set of states greater than $s$ (i.e., $\uparrow s=
\{ s' \; | \; s \preceq s'\}$). With abuse of notation we will denote with $\pred{\_}$
also its natural extension to sets of states.
%a finite basis for the set of immediate predecessors of states greater than $s$,
%i.e., $\pred{s}$ is a finite set of states such that
%$\{s'\ |\ s'\lred{} s'' \succeq s\} = \{s'\ |\ \exists s''\in \pred{s} \text{ s.t. } s''\preceq s'\}$.

According to the theory of well-structured transition systems~\cite{abdulla:96,Finkel:2001},
we have that several properties are decidable for such transition systems
(under some conditions discussed below).
We will consider the following properties. 
\begin{defi}
\label{def.otherproblems}
%let the \emph{termination problem}
%be to decide, given an initial configuration
%$\State$, wether there exists no infinite computation starting from $\State$.
%The \emph{control-state maintainability problem} is to decide, given an initial configuration
%$\State$ and a finite set of states $\mathbb{S} = \{ \State_1, \cdots , \State_n\}$,
%whether there exists a computation starting from $\State$
%where all states cover one of the $\State_i$.
%%all the computations starting from $\State$ eventually traverse 
%%a state in $S$. 
%
%The \emph{inevitability problem} is the dual, namely deciding
%whether  all the computations starting from $\State$ eventually visit
%% traverse 
%a state not covering one in $\mathbb{S}$.
%
Consider a well-structured transition system $({\cal S}, \lred{}, \preceq)$.
Given $\State \in {\cal S}$ the \emph{termination problem} is to 
decide whether $\State$ has an infinite computation;
the \emph{control-state reachability} problem is to decide,
given $\StateT \in {\cal S}$, whether there is $\StateT' \succeq \StateT$ 
such that $\State \lred{}^* \StateT'$.
%state $\State$ of a well-structured transition system 
%The \emph{process reachability problem} is to decide, given one state ${\State}$ and 
%one process $P$
%${\StateT}$ of a well-structured transition system with well-quasi ordering
%$\preceq$,
%whether there is $\StateT' \succeq \StateT$ such that $\State \lred{}^* \StateT'$.
\end{defi}
In well-structured transition systems termination is decidable when
the transition relation $\lred{}$ and the ordering $\preceq$ are
effectively computable. When it is also possible to effectively compute
a finite-basis for the set of states $\pred{\uparrow s}$ we have that 
control-state reachability is decidable as well. 

In the following we assume given an actor program with its 
main process and its set of actor class definitions.
The first relation we convey is $\eqdot$ that relates renamings,
ranged over by $\rho$, $\rho'$, $\ldots$
that are functions mapping variables \emph{that are
not free in the main process} into either actor names or variables.
% \emph{that are
%not free in the main process}.
%
%$\simeq$ on terms $m(x_1, \cdots, x_k)$ and on
%processes as follows. To this aim, let $\rho$ and $\rho'$ be renaming 
%that map variables to either variables or actor names) and 
Let 
{\small
\[
\begin{array}{lrll}
\rho \eqdot \rho' \quad \eqdef \quad & \mbox{ for every } x,y: 
&(i) \quad  
\rho(x) = \rho(y) \quad & % \mbox{ if and only if} \quad
\begin{array}{l}
\mbox{if and only if $\quad$ $\rho'(x) = \rho'(y)$} %\\
%\mbox{in the main process and $\rho'(x) = \rho'(y)$}
\end{array}
%\quad with \; \rho(x) \; \mbox{and} \; \rho'(x) \;
%\mbox{ not free in the main process}
\vspace{2mm}\\
&&(ii) \quad \rho(x) = \rho'(x) \quad & 
\begin{array}{l}
\mbox{if $\rho(x)$ or $\rho'(x)$ is an actor name or}\\
\mbox{a free variable of the main process}
\end{array}
%\\
%(iii) & \rho(x) = \rho'(x) \quad \mbox{ if $\rho(x)$ or $\rho'(x)$ is
%a free variable of the main process}
\end{array}
\] }
Namely, two renamings are in the relation $\eqdot$ if they identify the 
same variables, regardless the value they associate when such a value is a 
variable.
For example, $[x \mapsto y, y \mapsto z] \eqdot [x \mapsto x, y \mapsto z]$
and $[x \mapsto y, y \mapsto y, z \mapsto A] \eqdot [x \mapsto x', y \mapsto x',
z \mapsto A]$. However $[x \mapsto y, y \mapsto z] \not \eqdot [x \mapsto x, y \mapsto x]$ and $[x \mapsto A] \not \eqdot [x \mapsto B]$.
In general, if $\rho$ and $\rho'$ are injective renamings that 
always return variables then $\rho \eqdot \rho'$. 
The requirements of $\eqdot$ are stronger for actor names:
in this case the two renamings should be identical.
%
%We also 
By definition, renamings in relation according to $\eqdot$
never apply to free variables of the main process.
% and 
%never return free variables of the main processes.
This because these variables are possibly stored in fields of actors and their
renaming might change the behaviours of actors in a way that breaks the 
upward compatibility of the following relation $\preceq$ and $\lred{}$
(\emph{c.f.} proof of Theorem~\ref{thm.decidablestatelessandfinite}, part 
{\bf (2)}). From this also follows that the above renamings \emph{do not change the
main process} (because they do not apply to its free variables).

Let $\dom{\rho}$ be the domain of the renaming $\rho$. 
We denote by $P\rho$ the result of $P \subst{\rho(\wt{x})}{\wt{x}}$, where 
$\wt{x} = x_1, \cdots, x_n$ is a tuple containing the variables in $\dom{\rho}$ 
(without repetitions) and $\rho(\wt{x}) = \rho(x_1), \cdots, \rho(x_n)$.
%We also let $\rho \circ \varphi$ be the state
%\[
%(\rho \circ \varphi)(\f) = \left\{ 
%		\begin{array}{ll}
%		\rho(\varphi(\f)) & {\rm if} \; \varphi(\f) \in \dom{\rho}
%		\\
%		\varphi(\f) & {\rm otherwise}
%		\end{array} \right.
%\]
%

Next, let $\simeq$ be the least relation on terms $m(U_1, \cdots , U_n)$ %, on states,
and on processes such that
{\small
\[
\bigfract{ \rho \eqdot \rho'
	}{
	m(\rho(x_1), \cdots, \rho(x_k)) \simeq m(\rho'(x_1), \cdots, 
	\rho'(x_k))
	}
%\qquad \qquad
%\bigfract{ \rho \eqdot \rho'
%	}{
%	\rho \circ \varphi \simeq \rho' \circ \varphi
%	}
\qquad \qquad
\bigfract{ \rho \eqdot \rho'
	}{
	P\rho \simeq P\rho'
	}
\]
}
For example, it is easy to verify that $m(x,y) \simeq m(x',y')$ and that
$[x = A] y \invk m(x,A,y) \simeq 
[z = A] y' \invk m(z,A,y')$. On the contrary 
$[x = A] B \invk m(x,A,B) \not \simeq 
[z = A] y' \invk m(z,A,y')$. The rationale behind $\simeq$ is that it
identifies  processes that ``behave in similar ways'', namely they
enqueue ``similar invocations'' in the same actor queue. Method invocations
$m(U_1, \cdots , U_n)$ of a given actor 
are identified if the processes they trigger
 ``behave in similar ways''.
%We notice that, if $\rho$ and $\rho'$ are injective renamings then 
%$P\rho \simeq P\rho'$. Similarly if $\rho$ and $\rho'$ are constant
%renamings mapping variables in a same actor name.
%
%Let $\simeq$ be the relation on terms $m(x_1, \cdots, x_\ell)$ 
%defined as follows:
%\[
%\bigfract{(x_i = x_j \qquad \mbox{if and only if} \qquad y_i = y_j)^{i,j \in 1..\ell}
%	}{
%	m(x_1, \cdots, x_\ell) \simeq m(y_1, \cdots, y_\ell) 
%	}
%\]

\begin{lem}
\label{prop.finteterms}
Let $T$ be either a method invocation $m(U_1, \cdots , U_n)$ %, or a state, 
or a process of a program in {\actba} (and therefore in {\actroba}).
Let ${\cal T} = \{ T \rho_1, T\rho_2, T\rho_3, \cdots \}$ be
such that $i \neq j$ implies $T \rho_i \not \simeq T \rho_j$. 
Then ${\cal T}$ is finite.
\end{lem}

\ifconf
\else
\proof
 We demonstrate the lemma for processes, the argument is similar for 
method invocations. So, let  $P$ be a process. It is possible to count the 
number of renamings $\rho$ on $\free{P}$ that are different according to
$\eqdot$. In fact, the values of renamings on variables that are 
different from  $\free{P}$ do not play any role in the definition of ${\cal T}$.

The basic remark is that a renaming $\rho$ generates a \emph{partition} of the
set $\free{P}$: two variables $x$ and $y$ are in the same partition if and only
if $\rho(x) = \rho(y)$. If we restrict to renamings that map variables to
variables (and not actor names), then they are different according to 
$\eqdot$ if they
yield different partitions. The number of such renaming is the 
\emph{Bell number} of the cardinality of $\free{P}$, let it be
${\tt Bell}(\kappa)$, where $\kappa$ is the cardinality of $\free{P}$. In 
addition, in our case, renamings may map a variable 
to an actor name into a finite set $\{A_1, \cdots, A_\ell\}$. In this case
the identity of the actor name is relevant. If $\kappa \geq \ell$ then 
$((\combinator{\kappa}{\ell}) \times \ell! +1)\times {\tt Bell}(\kappa)$ is an upper bound to the different renamings according $\eqdot$. If $\kappa < \ell$ then
the upper bound is $(\ell! / \kappa! +1 )\times {\tt Bell}(\kappa)$. In any case
the number of different renamings according to $\eqdot$ is finite.

Henceforth the set ${\cal T}$ is finite as well.
\qed
\fi

The well-quasi-ordering relation on configurations relies on an (almost standard) 
\emph{embedding relation} $\leq$ on queues (except the part about $\simeq$, it 
is the one in~\cite{Finkel:2001}):
{\small
\[
\bigfract{
	{\it there \; exist} \; i_1 < i_2 < \cdots < i_k \leq h \; 
	{\it such \; that, \; for} \; j \in 1..k,  \; \; m_j(\wt{U_j}) 
	\simeq n_{i_j}(\wt{V}_{i_j})
	}{
	m_1(\wt{U_1}) \ldots m_k(\wt{U_k}) \leq n_1(\wt{V_1}) \ldots n_h(\wt{V_h})
	}
\]
}
Then, let
{\small
\[
\bigfract{P_i \simeq P_i'\quad  \mbox{\rm and}  \quad q_i \leq q_i'\quad \mbox{\rm for } {i \in 1..\ell}
	}{
	 A_1 \triangleright (P_1, \varphi_1, q_1), \cdots , A_\ell \triangleright (P_\ell, \varphi_\ell, q_\ell) 
 \; \preceq \; A_1 \triangleright (P_1', \varphi_1, q_1'), \cdots , A_\ell \triangleright 
 (P_\ell',\varphi_\ell, q_\ell')
 }
\]
}
It is worth noticing that the relation $\preceq$ constraints corresponding
elements $A \triangleright (P, \varphi, q)$ and $A \triangleright 
(P', \varphi, q')$ to have the same states. In fact these states are defined by
the main process using either its free variables or the actor names that it
has created. For this reason there are finitely many of them and the relation
$\preceq$ is parametric with respect to them.

\begin{thm}
\label{thm.decidablestatelessandfinite}
Let $({\cal S},\lred{})$ be a transition system of a program
of {\actroba}. 
Then $({\cal S}, \lred{}, \preceq)$ is a well-structured transition system.
\end{thm}

\ifconf
\else
\proof
{\bf (1)} \emph{$\preceq$ is a well-quasi-ordering}. 
%It is easy to prove that $\preceq$ is a quasi-ordering.
To prove that $\preceq$ is a well-quasi-ordering, we reason by contradiction.
Let $\State_1, \State_2, \State_3, \cdots$ be 
an infinite sequence of states in ${\cal S}$ such that, for every $i < j$,  $\State_i \not \preceq 
\State_j$. 
%It is easy to verify that without loss of generality we may assume that
%the sequence is such that for every $i < j$,  $\State_i \not \succeq
%\State_j$ (all the states are incomparable).
%Note that an infinite strictly decreasing sequence,
%i.e., for every $i$,  $\State_i \succeq \State_{i+1}$ and 
%$\State_i \neq \State_{i+1}$  is not possible.
Let $y_1$,$\cdots$,$y_m$ be a sequence of variables not free in the main process.
Consider:
\[
\begin{array}{ll}
{\tt subterms}(\adef{C}) = &
\{ P \quad | \quad \mbox{there exists a method } m \mbox{ s.t. }
P \mbox{ is a subterm of } \adef{C}.m(\wt{x})\}\ \cup
\\
& 
\{m(y_1, \cdots ,y_g)\quad | \quad
\mbox{there exists a method } m \mbox{ s.t. }
\adef{C}.m(\wt{x}) \mbox{ with  $|\wt{x}|=g$} \}
% \mbox{ and $\wt{x}$ has length $j$}
\end{array}
\]
The set ${\tt subterms}(\adef{C})$ is finite, but all the --possibly infinitely many-- processes
that can be executed
(or the messages that can be received)  
by an actor of class $\adef{C}$ are renamings $P\rho$ (or $m(\wt{y}) \rho$) 
of these terms.
Notice that by Lemma~\ref{prop.finteterms}, the number of terms $P \rho$ 
(and $m(\wt{y}) \rho$) which are different 
according to $\simeq$ is finite as well. 
%This means that 
% actor names $A$ have finitely
%many $P$ such that $A \triangleright (P,q)$ belong to some state of ${\cal S}$,
It is thus possible to extract 
a subsequence $\State_{i_1}, \State_{i_2}, \State_{i_3}, \cdots$ from 
$\State_1, \State_2, \State_3, \cdots$ such that, for every $A$, in the
elements $A \triangleright (P_{i_j}^A \rho_{i_j}, \varphi_{i_j}^A, q_{i_j}^A)$ and $A \triangleright (P_{i_k}^A\rho_{i_k}, \varphi_{i_k}^A, q_{i_k}^A)$
of $\State_{i_j}$ and $\State_{i_k}$, respectively, we have that $P_{i_j}^A\rho_{i_j} \simeq
P_{i_k}^A\rho_{i_k}$. Moreover, as we are considering {\actroba} the actor state cannot
be modified, hence $\varphi_{i_j}^A=\varphi_{i_k}^A$.

As we are considering {\actroba}, the set of actor is bound.
Let ${A_1}, \cdots , {A_\ell}$ be such actor names.
Due to the above arguments, the sequence $\State_{i_1}, \State_{i_2}, 
\State_{i_3}, \cdots$ may be represented as a sequence of tuples
% of length $\ell$
of queues:
%  (where ${A_1}, \cdots , {A_\ell}$ are the 
%actor names that are fixed because we are considering {\actroba}):
\[
(q_{i_1}^{A_1}, \cdots , q_{i_1}^{A_\ell}), \;
(q_{i_2}^{A_1}, \cdots , q_{i_2}^{A_\ell}),\;
(q_{i_3}^{A_1}, \cdots , q_{i_3}^{A_\ell}),\; \cdots
\]
such that $\State_{i_j} \preceq \State_{i_k}$ if and only if 
$(q_{i_j}^{A_1}, \cdots , q_{i_j}^{A_\ell}) \sqsubseteq^\ell
(q_{i_k}^{A_1}, \cdots , q_{i_k}^{A_\ell})$, where $\sqsubseteq^\ell$ is the 
coordinatewise order defined by
\[
(q_1, \cdots , q_\ell) \sqsubseteq^\ell
(q_1', \cdots , q_\ell') \quad \eqdef \quad 
\mbox{{\em for every h}} \; : \; q_h \leq q_h'
\]
($\leq$ is the above embedding relation).

We are finally reduced to an infinite sequence of tuples of queues 
such that every tuple cannot be in relation according to $\sqsubseteq^\ell$
with any of the subsequent ones.
%that are pairwise incomparable according to $\sqsubseteq^\ell$. 
This fact contradicts the 
\begin{itemize}
\item[] \emph{Higman's Lemma \cite{HigmanLemma}:
if $(X,\le)$ is a well-quasi-ordering and $(X^*, \le^*)$ is the set of finite 
$X$-sequences ordered by the embedding 
relation $\le^*$ defined using $\le$ as pointwise ordering, 
then $(X^*, \leq^*)$ is a  well-quasi-ordering.}
\end{itemize}
More precisely, the contradictions follows from the following
consequence of the Higman's Lemma:
\begin{itemize}
\item 
if $X$ is a finite set and $(X^*, \le)$ is the set of finite $X$-sequences 
ordered by the embedding 
relation, then $(X^*, \leq)$ is a  well-quasi-ordering.
\end{itemize}
and from the following statement
\begin{itemize}
\item 
if $(X,\leq)$ is a well-quasi-ordering then $(X^\ell, \leq^\ell)$ is a well-quasi-ordering.
\end{itemize}

\medskip

{\bf (2)} \emph{$\preceq$ is upward compatible with $\lred{}$}.
A state $\varphi$ is \emph{normed}, if, for every field $\f$, $\varphi(\f)$ is 
either a free variable in the main process or an actor name.
A configuration is \emph{normed} if the states of the actors are normed.
We observe that the initial configuration is 
 normed.
We also let $(\wt{E})\rho \simeq (\wt{E})\rho'$ whenever $\rho \eqdot \rho'$.

We first demonstrate that, if $\wt{E} \lleadsto{\varphi} \wt{U}\; ; \; \State$
with $\varphi$ normed, then

\begin{enumerate}[label=(exp-\roman*)]
\item%[\emph{(exp-i)}]
if  $\wt{E} \simeq \wt{E'}$ and $\State = \varnothing$ then
$\wt{E'} \lleadsto{\varphi} \wt{U'}\; ; \; \varnothing$ and 
$\wt{U} \simeq \wt{U'}$;

\item%[\emph{(exp-ii)}]
if $\wt{E}$ (respectively $P$) only contain  free variables in the main process 
and actor names
then $\wt{E} \simeq \wt{E}'$ (respectively $P \simeq P'$) implies
$\wt{E} = \wt{E}'$ (respectively $P =_\alpha P'$) 
and $\wt{E}  \lleadsto{\varnothing} \wt{U}\; 
; \; \State$ implies that $\wt{U}$ contain  free variables in the main process and actor names and $\State$ is normed.
\end{enumerate}

\noindent {(exp-i)} is proved by induction on the hight of the proof-tree of 
$\wt{E} \lleadsto{\varphi} \wt{U}\; ; \; \varnothing$. There are two basic cases:
(1) $E = U$ and (2) $E = \f$. As regards (1), $E' = U'$ and the property is
immediate by the hypothesis that $E \simeq E'$. As regards (2), $E' = \f$ because
$E \simeq E'$; henceforth the property (because $E'$ is evaluated in the state 
$\varphi$ as well). There is one inductive case (because the case of $\mathtt{new}$
is not possible, otherwise $\State$ cannot be empty), which is immediate.

{(exp-ii)} is an immediate consequence of the definition of $\eqdot$ and
$\lleadsto{\varnothing}$.

\bigskip

\noindent Let $\State_1 \lred{} \State_2$. We demonstrate that 

\begin{enumerate}[label=(\roman*)]
\item%[\emph{(i)}]
if $\State_1$ is normed then $\State_2$ is normed as well
(this means that the transition system $({\cal S},\lred{})$ of a program
of {\actroba} has normed configurations because the initial state is normed);

\item%[\emph{(ii)}]
if $\State_1 \preceq \State_1'$ then there exists $\State_1' \lred{}^* \State_2'$ such that $\State_2 \preceq \State_2'$.
\end{enumerate}
As regards \emph{(i)}, it follows by remarking that in programs of {\actroba},
there is no field update and the unique process
that may create states is the one of $\aleph$ (the main process). Then 
\emph{(i)} derives from the property \emph{(exp-ii)}.

As regards \emph{(ii)}, its proof is a case analysis on the proof-tree of 
$\State_1 \lred{} \State_2$ where the cases correspond to the unique rule appearing
in the tree that is not an instance of \rulename{context}. 
%, where 
% the proof-tree of $\lleadsto{\varphi}$ is also considered.
%
Let $\State_1 = A_1 \triangleright (P_1\rho_1, \varphi_1, q_1),
\cdots , A_\ell  \triangleright (P_\ell \rho_\ell, \varphi_\ell, q_\ell)$. Since $\State_1 
\preceq \State_1'$ then 
$\State_1' = A_1 \triangleright (P_1\rho_1', \varphi_1, q_1'),
\cdots ,  A_\ell  \triangleright (P_\ell\rho_\ell', \varphi_\ell,
q_\ell')$ such that, for every
$i$, 
$P_i \rho_i \eqdot P_i \rho_i'$ and $q_i \leq q_i'$. 
%The proof-tree of
%$\State_1 \lred{} \State_2$ contains exactly one instance of either
%\rulename{let} or \rulename{invk-s} or \rulename{inv} or \rulename{inst} or
%\rulename{match} or \rulename{mmatch} or \rulename{plus-l} or \rulename{plus-r},
%followed by a number of instances of \rulename{context}. 
The cases are discussed in order.

%
%The basic case is when the height is 0. There are three subcases, according to
%the transition is proved with an instance of either \rulename{inst} or 
%\rulename{plus-l} or \rulename{plus-r}.
%\begin{enumerate}
%\item
%$\State_1 \lred{} \State_2$ is an instance of \rulename{inst} where $\ell = 1$ 
%(rule \rulename{context} cannot be used). Therefore $A_1 = \aleph$ and this case
%is vacuous because the queue of $\aleph$ cannot contain any message.
%
%\item 
%$\State_1 \lred{} \State_2$ is an instance of either \rulename{plus-l} or
%\rulename{plus-r}. Immediate.
%\end{enumerate}
%
%The inductive cases are the following ones:
\begin{enumerate}
\item
$\State_1 \lred{} \State_2$ contains an instance of \rulename{let}, namely
\[A \triangleright (\letin{x}{E}{P}, 
\varphi, q), 
\; \lred{} \; A \triangleright (P\subst{U}{x}, \varphi, q ), \State_3\;.\] 
%(without loss of generality, we are assuming no name clash between $\rho$ 
%and the bound variable$x$) 
where $E \lleadsto{\varphi} U \, ; \, \State$. 
By  $\State_1 \preceq \State_1'$,
$\State_1'$ must contain $A \triangleright (\letin{x}{E'}{P'}, \varnothing,
\varepsilon)$ such that $\letin{x}{E}{P} \simeq \letin{x}{E'}{P'}$ (without loss of 
generality, we are assuming the two bound variables are the same) and $q \leq q'$.
There are two subcases: (1.1)
$A = \aleph$ and (1.2) $A \neq \aleph$. In (1.1), By \emph{(exp-ii)}, 
this is possible provided $E = E'$ and $P =_\alpha P'$. It is easy to verify that
$\State_1' \lred{} \State_2'$ and $\State_2 \preceq \State_2'$
because their unique difference with $\State_1$ and $\State_1'$
is due to the two processes $P$ and $P'$.
In (1.2), $\State_3 = \varnothing$ because no ${\tt new}$ can occur in $E$.
Additionally, by definition of $\simeq$, $E \simeq E'$ and $P \simeq P'$.
Let $E' \lleadsto{\varphi} U', \, \varnothing$. By \emph{(exp-i)} we have $U 
\simeq U'$ and it is easy to verify that $P\subst{U}{x} \simeq P'\subst{U'}{x}$.
Henceforth $\State_1' \lred{} \State_2'$ and $\State_2 \preceq \State_2'$
because their unique difference with $\State_1$ and $\State_1'$
is due to the two processes $P\subst{U}{x}$ and $P'\subst{U'}{x}$.

\item
$\State_1 \lred{} \State_2$ contains an instance of \rulename{invk-s}, namely
\[A \triangleright (A \invk m(\wt{E}) \prefix P, \varphi, q), 
\; \lred{} \; A \triangleright (P, \varphi, q \cdot  m(\wt{U}))\;.\]
Since $\State_1 \preceq \State_1'$ then $\State_1'$ contains $A \triangleright (A \invk m(\wt{E'}) \prefix P', \varphi, q')$ with
$\wt{E} \simeq \wt{E'}$, $P \simeq P'$, and $q \leq q'$. 
We observe that $A \neq \aleph$ and if $\wt{E} \lleadsto{\varphi} \wt{U}, \, \varnothing$ and $\wt{E'} \lleadsto{\varphi} \wt{U'}, \, \varnothing$ then
$\wt{U} \simeq \wt{U'}$ by \emph{(exp-i)}. Therefore $q \cdot  m(\wt{U}) \leq
q' \cdot m(\wt{U'})$  and $\State_1' \lred{} \State_2'$ with $\State_2 \preceq
\State_2'$ because their unique difference with $\State_1$ and $\State_1'$
is due to the two terms 
$A \triangleright (P, \varphi, q\cdot m(\wt{U}))$ and 
$A \triangleright (P', \varphi, q'\cdot m(\wt{U'}))$.

\item
$\State_1 \lred{} \State_2$ contains an instance of \rulename{invk},
namely 
\[A \triangleright (B \invk m(\wt{E}) \prefix P, \varphi, q) , 
B \triangleright (Q, \psi, p)
	 \; \lred{} \;
	A \triangleright (P, \varphi, q), B \triangleright (Q, \psi, p 
	\cdot m(\wt{U}) ), \, \State_3\;.\] 
There are two subcases: either $A = \aleph$ or $A \neq \aleph$. When $A = \aleph$
the proof is similar to the above case (1.1); when $A \neq \aleph$ the proof is 
similar to case (2).

\item
$\State_1 \lred{} \State_2$ contains an instance of \rulename{inst}, namely
\[A \triangleright (\pinull, \varphi, m(\wt{U}) \cdot q) \lred{} A \triangleright (
P\subst{A}{{\it this}}\subst{\wt{y'}}{\wt{y}}\subst{\wt{U}}{\wt{x}}, \varphi, q)\;,\]
where $\adef{C} \prefix m(\wt{x}) = P$, $\adef{C}$ being the class of $A$,
$\wt{y} = \free{P} \setminus \wt{x}$ and $\wt{y'} = \fresh{\wt{y}}$.
Therefore $\State_1 = A \triangleright (\pinull, m(\wt{U}) \cdot q),
\StateT_1$ and $\State_2 = A \triangleright (
P\subst{A}{{\it this}}\subst{\wt{y'}}{\wt{y}}\subst{\wt{U}}{\wt{x}}, q), \StateT_1$. 
Since $\State_1 
\preceq \State_1'$ then $\State_1' = A \triangleright (\pinull, \varphi, 
n_1(\wt{V_1}) \cdots
n_h(\wt{V_h}) \cdot  m(\wt{V}) \cdot q'), \StateT_1'$ and 
$m(\wt{U}) \simeq m(\wt{V})$ and $q \leq q'$ and $\StateT_1 \preceq 
\StateT_1'$. By the operational semantics rules, we get
$\State_1' \lred{}^* A \triangleright (\pinull, \varphi, 
m(\wt{V}) \cdot q' \cdot q''), \StateT_1''$ 
by performing transitions of the actor $A$, with $\StateT_1' \preceq \StateT_1''$ 
and, by definition, $q \leq
q' \cdot q''$. At this stage, we notice that $A \triangleright (\pinull, 
\varphi, m(\wt{V}) 
\cdot q' \cdot q''), \StateT_1'' \lred{} A \triangleright (P
\subst{A}{{\it this}}\subst{\wt{z}}{\wt{y}} \subst{\wt{V}}{\wt{x}}, \varphi, 
q' \cdot q''), \StateT_1''$. We notice that
$P\subst{A}{{\it this}}\subst{\wt{y'}}{\wt{y}}\subst{\wt{U}}{\wt{x}} 
= P \subst{A}{{\it this}} [ \wt{y} \mapsto \wt{y'}, \wt{x} \mapsto \wt{U}]$
and
$P\subst{A}{{\it this}}\subst{\wt{z}}{\wt{y}}\subst{\wt{V}}{\wt{x}}
= P \subst{A}{{\it this}} [ \wt{y} \mapsto \wt{z}, \wt{x} \mapsto \wt{V}]$
and
$[ \wt{y} \mapsto \wt{y'}, \wt{x} \mapsto \wt{U}] \eqdot 
[ \wt{y} \mapsto \wt{z}, \wt{x} \mapsto \wt{V}]$.
Therefore 
\[
P\subst{A}{{\it this}}\subst{\wt{z}}{\wt{y}}\subst{\wt{V}}{\wt{x}}
\simeq 
P\subst{A}{{\it this}}\subst{\wt{z}}{\wt{y}}\subst{\wt{V}}{\wt{x}}
\]
which implies that $\State_1' \lred{}^*\lred{} \State_2'$ and $\State_2  \preceq
\State_2'$
because their unique difference with $\State_1$ and $\State_1'$
is due to the two above processes.

\item
$\State_1 \lred{} \State_2$ contains an instance of \rulename{match}, namely
\[A \triangleright ([E=E'] P \ite Q , \varphi, q) 
	\; \lred{} \; A \triangleright (P , \varphi, q), \State_3\;.\]
We discuss the case $A \neq \aleph$ because the other one is similar to (1.1).
There are three subcases (5.1) both $E$ and $E'$ are variables;
(5.2) $E$ is a variable and $E'$ is a field; (5.3) $E$ and $E'$ are
both fields. In case (5.1), let $E = x = E'$. Since $\State_1 
\preceq \State_1'$, then $\State_1'$ must contain $A \triangleright 
([z = z] P' \ite Q' , \varphi, q')$ with $[x=x] P \ite Q \simeq 
[z= z] P' \ite Q'$ and $q \leq q'$. Therefore we may use \rulename{match} to derive
$\State_1'  \lred{} \State_2'$ with $\State_2 \preceq \State_2'$.
In case (5.2), let $E = U$ and $E' = \f$.
There are two subcases: (5.2.1) $U$ is a variable or (5.2.2) $U$ is an actor name.
In (5.2.1), $U$ has to be a free variable in the main process because we are
using \rulename{match} ($\f$ may contain either such variables or actor names, additionally, renamings never return free variables in the main process). Therefore, by 
$\State_1 
\preceq \State_1'$, we have that $\State_1'$ contains $A \triangleright 
([U = \f] P' \ite Q' , \varphi, q')$ with $[U=\f] P \ite Q \simeq 
[U= \f] P' \ite Q'$ and $q \leq q'$. The consequence is that
$\State_1'  \lred{} \State_2'$ with $\State_2 \preceq \State_2'$ 
because their unique difference with $\State_1$ and $\State_1'$
is due to the two either the pair of processes $P$, $P'$
or $Q$, $Q'$.
Similarly for (5.2.2). 
The case (5.3) is obvious. 

\item
$\State_1 \lred{} \State_2$ contains an instance of \rulename{mmatch}. Similar to (5).

\item
$\State_1 \lred{} \State_2$ contains an instance of \rulename{plus-l} or of 
\rulename{plus-r}. Straightforward.
\qed
\end{enumerate}
\fi

We notice that the well-structured transition system $({\cal S}, \lred{}, \preceq)$ 
%has transitive and stuttering compatibility (see \cite{Finkel:2001}, pp 9, 10). Additionally, $({\cal S}, \lred{}, \preceq)$ 
has decidable algorithms for computing $\preceq$ and for computing the next states. 
Then decidability of termination directly follows from the above mentioned
results of the theory of well-structured transition systems that we have
previously recalled.
\ifcamera
Then decidability of termination follows.
%directly from Theorems 4.6 in~\cite{Finkel:2001}.
\else
%Then a number of
%decidability properties follow directly by Theorems 4.6
%and 4.8 in~\cite{Finkel:2001}. 
%To this aim, we integrate 
%Definition~\ref{def.termandreach} with the following. %problems.
\fi

\ifcamera
\else
%\begin{defi}
%\label{def.otherproblems}
%%let the \emph{termination problem}
%%be to decide, given an initial configuration
%%$\State$, wether there exists no infinite computation starting from $\State$.
%%The \emph{control-state maintainability problem} is to decide, given an initial configuration
%%$\State$ and a finite set of states $\mathbb{S} = \{ \State_1, \cdots , \State_n\}$,
%%whether there exists a computation starting from $\State$
%%where all states cover one of the $\State_i$.
%%%all the computations starting from $\State$ eventually traverse 
%%%a state in $S$. 
%%
%%The \emph{inevitability problem} is the dual, namely deciding
%%whether  all the computations starting from $\State$ eventually visit
%%% traverse 
%%a state not covering one in $\mathbb{S}$.
%%
%Consider a well-structured transition system $({\cal S}, \lred{}, \preceq)$.
%Given $\State \in {\cal S}$ the \emph{termination problem} is to 
%decide whether $\State$ has an infinite computation;
%the \emph{control-state reachability} problem is to decide,
%given $\StateT \in {\cal S}$, whether there is $\StateT' \succeq \StateT$ 
%such that $\State \lred{}^* \StateT'$.
%%state $\State$ of a well-structured transition system 
%%The \emph{process reachability problem} is to decide, given one state ${\State}$ and 
%%one process $P$
%%${\StateT}$ of a well-structured transition system with well-quasi ordering
%%$\preceq$,
%%whether there is $\StateT' \succeq \StateT$ such that $\State \lred{}^* \StateT'$.
%\end{defi}
\fi

\begin{thm}
\ifcamera
In {\actroba} termination is decidable.
\else
In programs of {\actroba}
%with finitely many actors and read-only fields, 
the termination
%the control-state reachability problem, and the inevitability problem are 
problem is decidable.
\fi
\end{thm}

%As discussed in Section~\ref{sec.thelanguage}, the transition systems of the
%actor language are not finite branching. This is also the case for
%programs in {\actroba} (due to the presence of
%fresh variables in method body instantiations). However, in this case, the sets
%$\Succ{\State}$ and $\pred{\State}$ are finite if we reason up-to the well-quasi
%ordering relation $\preceq$.
We now move to the definition of an appropriate algorithm for the
computation of a finite basis for the predecessors of a given configuration,
so to conclude also the decidability of control-state reachability.

\begin{lem}
\label{lem.pred}
Let $({\cal S}, \lred{}, \preceq)$ be a well-structured transition system 
of a program in {\actroba}, and let $\State \in
{\cal S}$. Then there is a finite set ${\cal X} \subseteq \pred{\State}$ 
such that, for every $\State' \in \pred{\State}$, there is $\StateT \in {\cal X}$
with $\StateT \preceq \State'$. ${\cal X}$ can be effectively computed.
\end{lem}

\ifconf
\else

\proof
We show how to compute ${\cal X}$. Let $\State = 
A \triangleright (P,\varphi, q), \State'$. The \emph{predecessor processes} of $P$ are the following ones:
(\emph{i}) $\letin{x}{E}{P'}$, with $P = P' \subst{\wt{U}}{\wt{x}}$, for some 
$\wt{U}$ and some ${\wt{x}}$;
(\emph{ii}) $ x \invk m(E_1, \cdots , E_n) \prefix P$;
(\emph{iii}) $[U=U] \; P \ite Q$;
(\emph{iv}) $[U=V] \; Q \ite P$;
(\emph{v}) $P + Q$;
(\emph{vi}) $Q+P$;
(\emph{vii}) $P$ is an instance of a method body of the actor class of $A$.
If $A$ is of actor class $\adef{C}$ then we take all the method bodies of
$\adef{C}$ with a suffix matching one of the cases (\emph{i})--(\emph{vi}) 
above (in this case, the expressions in (\emph{ii}) are either variables or
actor names). If $A = \aleph$ then we look for a matching suffix of the 
main process. The above six cases are demonstrated 
in the presence of such suffixes.

We only discuss case (\emph{i}), the other ones are similar.
In case (\emph{i}), if $A$ is of actor class $\adef{C}$, then 
$E = y$, for some $y$. If $x \in \free{P'}$ then  ${\cal X}$ contains the configuration
$A \triangleright (\letin{x}{y}{P'}, \varphi, q), \State'$ with 
$P = P' \subst{y}{x}$. Otherwise ${\cal X}$ contains the configuration
$A \triangleright (\letin{x}{z}{P'},\varphi, q), \State'$, for $z
\in \free{P'}$ and for a unique $z \notin \free{P'}$.
When $A = \aleph$ then $E$ may be $\newact{C}$ (orherwise the 
argument is as before). 
If $x \in \free{P'}$ and $\State' = A' \triangleright (\pinull, \varphi, 
\varepsilon), \State''$ with $A' \in \adef{C}$
then  ${\cal X}$ contains the configuration
$A \triangleright (\letin{x}{\newact{C}}{P'},q), \State''$
(and this for every possible $A' \in \adef{C}$ such that
$A' \triangleright (\pinull, 
\varepsilon)$ is in $\State'$).
\qed
\fi

Lemma~\ref{lem.pred} and the above mentioned results
on well-structured transition systems
allow us to decide the 
\ifcamera
 \emph{control-state
reachability problem}:
given two states ${\State}$ and ${\StateT}$ of a well-structured transition system with well-quasi-ordering
$\preceq$, decide
whether there is $\StateT' \succeq \StateT$ such that $\State \lred{}^* \StateT'$.
\else
 \emph{control-state
reachability problem}.
\fi

\begin{thm}
\label{thm.cs-reachability}
\ifcamera
In {\actroba} process reachability is decidable.
\else
In programs of {\actroba}
%with finitely many actors and read-only fields, 
the control-state reachability problem is decidable.
\fi
\end{thm}

\ifconf
\else
\proof
Let $\uparrow \State =
\{ \State' \in {\cal S} \; | \; \State \preceq \State'\}$. Let also
$\pred{\uparrow \State} = \{ \StateT \; | \; \StateT \lred{} \State'
\; \mbox{and} \; \State' \succeq \State\}$. 
By definition of $\preceq$, $\pred{\uparrow \State} \subseteq 
\pred{\State}$. Therefore 
$\uparrow \! \pred{\uparrow \State} \subseteq \; \uparrow \! \pred{\State} \subseteq
\; \uparrow \! {\cal X}$, where ${\cal X}$ is the finite set of Lemma~\ref{lem.pred}
that is effectively computable. 
%The theorem follows from Theorem 3.6 
%in~\cite{Finkel:2001}.
\qed
\fi

Next we discuss the process reachability problem -- 
see Definition~\ref{def.termandreach} -- in {\actroba}. To this aim, we use a 
simpler version of the (classical) diamond property.

\begin{prop}
\label{prop:simpleDiamond}
Let $({\cal S},\lred{})$ be a transition system of a program
of {\actroba} and let
$\aleph \; \triangleright \; (P, \varnothing, \varepsilon),\State
\lred{} \aleph \; \triangleright \; (P, \varnothing, \varepsilon),\State'$
($\aleph$ does not move) and 
$\aleph \; \triangleright \; (P, \varnothing, \varepsilon),\State'
\lred{} \aleph \; \triangleright \; (P', \varnothing, \varepsilon),\State''$
with $P' \neq P$ ($\aleph$ moves).
Then there exists $\State'''$ such that
$\aleph \; \triangleright \; (P, \varnothing, \varepsilon),\State
\lred{} \aleph \; \triangleright \; (P', \varnothing, \varepsilon),\State'''
\lred{} \aleph \; \triangleright \; (P', \varnothing, \varepsilon),\State''$.
\end{prop}
It is worth noticing that the language {\actor} also owns the more classical 
diamond property:
if in a configuration there are two transitions inferred by two
distinct actors, then it is possible to perform them in any order
reaching the same configurations \emph{up-to bijective renaming}.
We omit the formalization of this property since it is not needed
in the rest of the paper.

\begin{cor}
The process reachability problem
is decidable in {\actroba}.
\end{cor}

\begin{proof}
In order to verify whether a configuration 
$A \triangleright (P', \varphi, q),\State$
is reachable with $P'$ equal to $P$ up-to renaming of variables and \emph{actor names},
we  proceed as follows.

First, by Proposition~\ref{prop:simpleDiamond}
it is not restrictive to consider the set of configurations ${\cal T}$
reachable by completely executing the actor $\aleph$ only.
The cardinality of ${\cal T}$ is bounded by $2^k$, where $k$ is the maximal nesting
of $+$ in the main process. 
If one of the processes in the configurations reached by executing $\aleph$ is
equal to $P$, up-to renaming of variables and \emph{actor names}, then we are done.
Otherwise, let $\wt{u}$ be the free variables in the main process.
For each of 
$\StateT=\aleph \triangleright (\pinull , \varnothing, \varepsilon),
A_1 \triangleright (\pinull, \varphi_1, q_1), \cdots , A_\ell \triangleright (\pinull, \varphi_\ell, q_\ell)$ in ${\cal T}$,
we check control-state reachability
from 
$\StateT$ to at least one of the states
in the following finite set:
\[
\begin{array}{l}
%{\cal T}' & = & 
\{\ \aleph \triangleright (\pinull , \varnothing, \varepsilon),A_1 \triangleright (Q_1\subst{A_1}{{\it this}}\subst{\wt{z_1}}{\wt{y_1}}\subst{\wt{U_1}}{\wt{x_1}} , \varphi_1, \varepsilon), \cdots , A_\ell \triangleright (
Q_\ell\subst{A_\ell}{{\it this}}\subst{\wt{z_\ell}}{\wt{y_\ell}}\subst{\wt{U_\ell}}{\wt{x_\ell}}
, \varphi_\ell, \varepsilon) 
\\
%& & 
\qquad \mid \quad
\mbox{\em for every $1\leq i\leq \ell$, $Q_i$ is a suffix of the body of ${\tt m}_i$ in 
$\adef{C}_i$, where $A_i \in \adef{C}_i$,}
\\
%& & 
\qquad \qquad\qquad
\mbox{\em formal parameters and 
free variables of ${\tt m}_i$ are $\wt{x_i}$ and $\wt{y_i}$}
\\
%& & 
\qquad\qquad\qquad
\mbox{\em $\wt{U_i}$ is a tuple in $\{A_1,\ldots,A_\ell, \wt{u}, \wt{z}\}$
\quad ($\wt{z}, \wt{z_1}, \cdots, \wt{z_\ell}$ are fresh) }
 \\
%& & 
\qquad\qquad\qquad
\mbox{\em there exists $1\leq j \leq \ell$ such that $Q_j$ is
equal to $P$ up-to renaming}
\ \}
\end{array}
\]
Then the corollary follows by Theorem~\ref{thm.cs-reachability}.
\end{proof}

We conclude this section by recalling that 
%the hypotheses of
%Theorem~\ref{thm.decidablestatelessandfinite} can be extended to actors 
%with read-only fields which are initialized upon activation (details are 
%straightforward and therefore ommitted).
%On the other hamd, 
we have already proved
the undecidability of termination in programs with unboundedly many 
actors and read-only fields.
Note that 
%This undecidability result is consistent with the fact that 
if we remove from {\actroba} the constraint on bounded actor names
then  the relation $\preceq$ is no longer a well-quasi-ordering.
Consider, for instance, an actor (with empty state) having a method 
that first creates a new instance of the same class and then invokes 
on this new instance the same method. Among the reachable configurations
it is possible to select a sequence 
$\State_1, \State_2, \State_3, \cdots$ such that
the configuration $\State_n$ is defined as follows:
%\[
%\begin{array}{r@{\qquad}l}
%\State_n \; \eqdef & A_0 \triangleright (\pinull,\varnothing, m(A_{n-1},A_1)) \; , \\
%& A_1 \triangleright (\pinull,\varnothing, m(A_{n},A_2)) \; ,
%\\
%&
%\bigcup_{i\in 2..n-1} A_i \triangleright (\pinull,\varnothing, m(A_{i-2},A_{i+1}))
%\; ,
%\\
%& A_n \triangleright (\pinull,\varnothing, m(A_{n-2},A_0))
%\end{array}
%\]
\[
\begin{array}{r@{\qquad}l}
\State_n \; \eqdef & A_1 \triangleright (\pinull, \varnothing, \varepsilon) \; , 
\cdots ,  A_n \triangleright (\pinull, \varnothing,  \varepsilon)
\end{array}
\]
It is easy to see that, 
%The infinite sequence $\State_1, \State_2, \State_3, \cdots$ is 
% such that, 
for every $i < j$,  $\State_i \not \preceq 
\State_j$. 

\section{Decidability results for {\actsl}}
\label{sec.stateless}

\newcommand{\name}[1]{{\it name}(#1)}
\newcommand{\abseval}[1]{\lleadsto{#1}_\alpha}
\newcommand{\absred}[1]{\stackrel{#1}{\longrightarrow_{\mathsf a}}}
\newcommand{\abst}[1]{\Omega(#1)}
\newcommand{\rearrange}{\bowtie}
\newcommand{\shuffle}[2]{\it shuffle(#1,#2)}
\newcommand{\shuffleq}[2]{\it [\!\![ #1 ]\!\!]_{#2}}
\newcommand{\mulset}{\mathcal M}
\newcommand{\trans}[2]{\ensuremath{\xrightarrow{#1}}}
\newcommand{\abtrans}[2]{\ensuremath{\xrightarrow{#1}}_{\mathsf a}}
\newcommand{\code}[1]{\ensuremath{\mathcal{Q}(#1)}}

We prove that in {\actsl}
%stateless programs,
termination and process reachability are decidable, too.
As discussed at the end of Section~\ref{sec.decidability},
the ordering defined for {\actroba} is not appropriate for {\actsl}
because in the latter it is possible to dynamically produce unboundedly many actors.
%as it is the case for {\actsl}.
%we have not succeeded in demonstrating these decidability results
%by patching the definition of $\preceq$ in Section~\ref{sec.decidability}.
%%
%The reason is that {\actsl} programs may produce unboundedly many actor names.
%
Therefore, in order to compute an upper bound to
the instances of method bodies, which is the basic argument for the model of
Section~\ref{sec.decidability} to be a well-structured transition system, 
we need to abstract
from the identity of these names -- as we have done with variables.
However, in case of actor names, the abstractions we have devised all break the
delivering of messages. 
Therefore we decided to apply our arguments to an abstraction of the operational model where the delivery of messages is inexact: it may be enqueued in every actor of the same class. 
Yet, this abstract model allows us
to derive decidability of termination and process reachability for the original language. 

\begin{table}[t]
The \emph{decorated evaluation relation} $E \lleadsto{}_d U \; ; \; \State$: 
\[
\begin{array}{c}
U \lleadsto{}_d U \; ; \; \varnothing
\qquad 
\bigfract{ \wt{E} \lleadsto{}_d \wt{U}\; ; \; \State 
	\quad A = \fresh{\adef{C}}
	}{\newact{\adef{C}}(\wt{E}) \lleadsto{}_d A\; ; \; A \triangleright (\varepsilon:\pinull, 
	\varnothing, \varepsilon), \State
	}
\\
\\
\bigfract{ E_i \lleadsto{}_d U_i \; ; \; \State_i, \quad \mbox{\rm for} \quad {i \in 1..n}
	}{
	E_1, \cdots , E_n \lleadsto{}_d U_1, \cdots , U_n \; ; \; \State_1, 
	\cdots , \State_n
	}
\end{array}
\]
The \emph{decorated transition relation} $\State \lred{\sigma} \State'$:
%\begin{figure}[h]
\[
\begin{array}{c}
\mathrule{let$_d$}{E \lleadsto{}_d U \; ; \;\State
	}{ 
	\begin{array}{l}
	A \triangleright (\sigma\cdot n:\letin{x}{E}{P}, \varnothing, q) 
	\; \lred{\sigma\cdot n+1} \;
	 A \triangleright (\sigma\cdot n+1:P\subst{U}{x}, \varnothing, q), \State
	 \end{array}
	} 
%\\
\\
\mathrule{invk-s$_d$}{ \wt{E} \lleadsto{}_d \wt{U} \; ; \;\State
	}{
	\begin{array}{l}
 	A \triangleright (\sigma\cdot n:A \invk m(\wt{E}) \prefix P, \varnothing, q)
	\\
	\qquad  \qquad \lred{\sigma\cdot n+1|m(\wt{U},A)} \quad 
	A \triangleright (\sigma\cdot n+1:P, \varnothing,q \cdot  m(\wt{U},\sigma\cdot n+1)) , \State
	\end{array}
	}
\\
%\\
\qquad
\mathrule{invk$_d$}{ \wt{E} \lleadsto{}_d \wt{U} \; ; \; \State
	}{
	\begin{array}{l}
  	A \triangleright (\sigma\cdot n:A' \invk m(\wt{E}) \prefix P, \varnothing, q) , 
	A' \triangleright (\sigma':P', \varnothing, q')
	 \\
	\qquad  \qquad \lred{\sigma\cdot n+1|m(\wt{U},A')} \quad
	A \triangleright (\sigma\cdot n+1:P,\varnothing,q), A' \triangleright (\sigma':P',\varnothing, q' \cdot m
	(\wt{U},\sigma\cdot n+1)) , \State
	\end{array}
	}%\hfill (A \not = B) 
%\\
\\
\mathrule{inst$_d$}{A \in \adef{C} \quad \adef{C} \prefix m(\wt{x}) = P 
	\quad 
	\wt{y} = \free{P} \setminus \wt{x}
	\quad
 	\wt{y'} = \fresh{\wt{y}}
	}{
	A \triangleright (\sigma':\pinull, \varnothing, m(\wt{U},\sigma) \cdot q)
	\; \lred{\sigma\cdot 0} \; 
	A \triangleright(\sigma\cdot 0:P\subst{A}{{\it this}}\subst{\wt{y'}}{\wt{y}} 
	\subst{\wt{U}}{\wt{x}}, \varnothing, q)
	}
%\\ 
\\
\mathrule{match$_d$}{E , E' \lleadsto{}_d U,U \; ; \; \State
	}{
	A \triangleright (\sigma\cdot n:[E=E'] P \ite Q , \varnothing, q)
	\; \lred{\sigma\cdot n+1} \; A \triangleright (\sigma\cdot n+1:P , \varnothing, q) , \State
	}
\\
%\\
\mathrule{mmatch$_d$}{E,E' \lleadsto{}_d U,V \; ; \; \State
	  \quad U \neq V 
	}{
	A \triangleright (\sigma\cdot n:[E=E'] P \ite Q , \varnothing, q)
	\; \lred{\sigma\cdot n+1} \; A \triangleright (\sigma\cdot n+1:Q , \varnothing, q) , \State
	}
%\\
\\
\mathax{plus-l$_d$}{A \triangleright (\sigma\cdot n:P+Q,\varnothing, q)
	\; \lred{\sigma\cdot n+1} \; A \triangleright (\sigma\cdot n+1:P,\varnothing, q)
%	\State'
%	}{
%	A \triangleright(P + Q, q) , \State \; \lred{\sigma} \; 
%	\State'
	}\\
%\qquad
\mathax{plus-r$_d$}{A \triangleright (\sigma\cdot n:P+Q,\varnothing, q)
	\; \lred{\sigma\cdot n+1} \; A \triangleright (\sigma\cdot n+1:Q,\varnothing, q)
%	\State'
%	}{
%	A \triangleright(Q + P, q) , \State \; \lred{\sigma} \; 
%	\State'
	}
\qquad
\mathrule{context$_d$}{\State 
	\; \lred{\sigma} \; 
	\State'
	}{
	\State , \State'' \; \lred{\sigma} \; 
	\State', \State''
	}
\\
\\
% \bigfract{S \equiv S' ~~~ S' \to S'' ~~~ S'' \equiv S'''}
%       {S \to S'''} 
% \hfill \equiv \mbox{ is associativity and commutativity of ``,''}
\end{array}
\]
\caption{\label{tab.decopsem} The decorated operational semantics of the language {\actsl}}
\end{table}

In order to formalize the correspondence between the concrete and the abstract 
operational semantics, we need to add decorations to processes and transitions
at the concrete level. Such decorations are used to keep track of the causal
dependencies among processes.
% (a process corresponding to a derivative of
%the instantiation of a method invocation causally depends on its predecessors
%and the process that issued such invocation).
The decorated syntax adds a sequence of natural numbers in front of the
process of an actor, namely, we use $A \triangleright (\sigma:P, \varnothing, q)$
where $\sigma$ has the following meaning: if $\sigma=\sigma'\cdot n$, then $\sigma'$ 
identifies the action of emission of the message that caused the method instantiation
from which $P$ was generated, and $n$ is a counter indicating that $P$ is actually 
generated by the method instantiation after $n$ steps.
Notice that for the main process executed by the actor $\aleph$ the sequence
$\sigma'$ is empty, and that when a method is instantiated the counter
$n$ is initialized to 0. 
In order to transfer the sequence from the message emitter to the 
receiving actor, we add $\sigma$ at the end of messages.
% in order to transfer this
%information from the message emitter to the actor in which it will be subsequently 
%instantiated. 
Namely, messages are now denoted with $m(\wt{U},\sigma)$.
The decorated operational semantics $\State \lred{\alpha} \State'$ 
is defined in Table~\ref{tab.decopsem}, where the label $\alpha$ can be
either a sequence $\sigma$ or a pair $\sigma|m(\wt{U},A)$ where the
second element identifies the message issued during the transition.
The decorated operational semantics increments the last number of the 
sequence of a process every time it performs an action, adds to messages
the current sequence of the emitter, and use the sequences
inside messages to initialize the sequence of the method instantiations
(by extending it with 0).

It is trivial to see that the operational semantics in Table~\ref{tab.decopsem}
and the decorated semantics coincide, in the sense that given
one configuration $\State_1$ of {\actsl} we have $\State_1 \lred{} \State_2$
if and only if there exist a label $\alpha$
and two decorations $\State'_1$ and $\State'_2$ of $\State_1$ and $\State_2$, respectively,
such that $\State'_1 \lred{\alpha} \State'_2$.

As discussed at the beginning of this Section, 
we need a more abstract semantics with inexact message deliveries.
This is obtained by 
changing the operational semantics in Table~\ref{tab.decopsem} by decoupling
the evaluation of the body of a method from
the actor name of that method.
%
%In this section we adapt the decidability
%results of the previous section to the
%case of unbounded actors, while
%keeping the restriction on stateless programs.
%We will show that it is still possible to resort
%to the theory of well-structured transition
%systems, but we have to move to an abstract
%semantics.
%
Let $\State \absred{\alpha} \State'$ be the \emph{abstract transition relation}
 defined as
$\State \lred{\alpha} \State'$ in Table~\ref{tab.decopsem} except the two
rules \rulename{invk-s$_d$} and \rulename{invk$_d$} for method invocation and the
rule \rulename{inst$_d$} for the instantiation of method bodies, which are replaced by 
those in Table~\ref{tab.abstrules}.
\begin{table}[t]
{\small
\[
\begin{array}{c}
\mathrule{invk-s$_a$}{ \wt{E} \lleadsto{}_d \wt{U} \; ; \;\State
\quad
A,A' \in {\adef{C}} 
	}{
 	A \triangleright (\sigma\cdot n : A' \invk m(\wt{E}) \prefix P, \varnothing, q) \; 
	\absred{\sigma\cdot n+1|m(\wt{U},A')} \; A \triangleright (\sigma\cdot n+1 : P, \varnothing,q \cdot  m(\wt{U},\sigma\cdot n+1,A')) , \State
	}
\\
\\
\mathrule{invk$_a$}{ \wt{E} \lleadsto{}_d \wt{U} \; ; \; \State
\qquad
A',A'' \in {\adef{C}} 
	}{
  	\begin{array}{l}
	A \triangleright (\sigma\cdot n : A' \invk m(\wt{E}) \prefix P, \varnothing, q) , 
	A'' \triangleright (\sigma':P', \varnothing, q') \\
	 \qquad\qquad \absred{\sigma\cdot n+1|m(\wt{U},A')} \;
	A \triangleright (\sigma\cdot n+1:P,\varnothing,q), A'' \triangleright (\sigma' : P',\varnothing, q' \cdot m
	(\wt{U},\sigma\cdot n+1, A')) , \State
	\end{array}
	}
\\
\\	
\mathrule{inst$_a$}{A \in \adef{C} \quad \adef{C} \prefix m(\wt{x}) = P 
	\quad 
	\wt{y} = \free{P} \setminus \wt{x}
	\quad
 	\wt{y'} = \fresh{\wt{y}}
	}{
	A \triangleright (\pinull, \varnothing, m(\wt{U}, \sigma, A') \cdot q)
	\; \absred{\sigma\cdot 0} \; 
	A \triangleright(\sigma\cdot 0 : P\subst{A'}{{\it this}}\subst{\wt{y'}}{\wt{y}} 
	\subst{\wt{U}}{\wt{x}}, \varnothing, q)
	}	
\end{array}
\]
}
\caption{\label{tab.abstrules} Abstract transition rules for method invocations and
instantiations}
\end{table}
In the abstract transition relation, a message is added in a queue
of an actor \emph{nondeterministically selected}
among those belonging to the
same class of the target actor. The item $m(\wt{U},\sigma)$ is enqueued 
with an additional argument -- the actor name of the target actor. 
This additional argument is used when the method body is instantiated. In fact
it replaces the variable ${\it this}$, thus making the execution
of a body invariant regardless the actor that actually performs it.

As an example, consider the task manager specified
in {\actsl} in the Example~\ref{ex:taskManager}.
Also under the abstract semantics $n$ distinct workers
are instantiated, but it is possible for two distinct
tasks to be delivered to the same worker.

We now introduce few notations:
\vspace{-1mm}
\begin{itemize}
\item[--]
Let $\abst{}$ be a map
from ``concrete'' to ``abstract'' configurations:
given a configuration $\State$, we denote with
$\abst{\State}$ the configuration
obtained from $\State$
by replacing each of its
actors $A \triangleright (\sigma:P, \varnothing, q)$
with $A \triangleright (\sigma:P, \varnothing, q')$
where $q'$ is obtained from $q$ by adding 
to each method invocation the %additional 
parameter $A$.
\item[--] Given a decorated configuration $\State$ and a label
$\alpha$, such that $\alpha=\sigma$ or $\alpha=\sigma|m(\wt{U},A)$,
we use $\proj{\State}{\alpha}$ to denote
the process decorated with $\sigma$ in $\State$:
$\proj{\State}{\alpha} = P$ if $\State$ contains the actor
$A \triangleright (\sigma:P, \varnothing, q)$, for some
$A$ and $q$.
\item[--] Let $\eqdot_{\mathsf a}$ be the following relation 
on variable renamings (not applied to variables that are free
in the main process) 
%such that
%We begin by redefining the relation $\simeq$.
%We introduce $\simeq_\alpha$ that identifies
%more terms with respect to $\simeq$
%as renaming of actor names is also allowed.
%Namely, let $\rho_\alpha$ and $\rho'_\alpha$ be renamings
%that map variables to either variables or actor names,
%and actor names to actor names, and let 
%\[
%\begin{array}{rl}
%\rho_\alpha \eqdot_\alpha \rho'_\alpha \quad \eqdef \ \ \ 
%\left \{
%\begin{array}{ll}
%\forall x,y.\ \ \  & \rho_\alpha(x) = \rho_\alpha(y) \quad \mbox{ if and only if } \quad 
%\rho'_\alpha(x) = \rho'_\alpha(y) \\
%\forall A. & \exists {\adef C} \quad s.t. 
%\quad A,\rho_\alpha(A) \in {\adef C} \\
%%\forall A,B.\ \ \ \  & \rho_\alpha(A) = \rho_\alpha(B) \quad \mbox{ if and only if } \quad
%%\rho'_\alpha(A) = \rho'_\alpha(B) 
%\end{array}
%\right .
%
%\end{array}
%\] 
\vspace{-3mm}
\[
\begin{array}{lrl}
\rho \eqdot_{\mathsf a} \rho' \quad \eqdef & \quad \mbox{ for every } x,y: &
\\
& (i) \quad &  
\rho(x) = \rho(y) \quad \mbox{ if and only if } \quad
\mbox{$\rho'(x) = \rho'(y)$}
\\
&(ii) 
%\quad & \rho(x) \mbox{ is a variable \quad if and only if \quad} \rho'(x) \mbox{ is a variable}
%\\
%&(iii) 
\quad & \rho(x) \in \adef{C} \quad \mbox{if and only if} \quad \rho'(x) \in \adef{C}\\
&(iii)
\quad & \rho(x) = \rho'(x) \quad 
\begin{array}{l}
\mbox{if $\rho(x)$ or $\rho'(x)$ is a free variable}\\
\mbox{of the main process}
\end{array}
\vspace{-1mm}
\end{array}
\] 
%We notice that, d
Differently from
the definition of $\eqdot$, %the relation 
$\eqdot_{\mathsf a}$ does not care of the 
identity of actor names (it is sufficient that
they belong to the same class). %We also notice that 
%Moreover, %the relation 
%$\eqdot_{\mathsf a}$ identifies two renamings
%that ``have matching types'', letting the type of variable being distinct from those
%of class actors.
\item[--]
Let $\simeq_{\mathsf a}$ be the relation defined as $\simeq$ in Section~\ref{sec.decidability}, with $\eqdot_{\mathsf a}$
instead of $\eqdot$. We extend it
to messages containing sequences and actor names as follows:
$m(\wt{U},\sigma,A) \simeq_{\mathsf a} m(\wt{U}',\sigma',A')$ iff $m(\wt{U}) \simeq_{\mathsf a} m(\wt{U}')$,
$\sigma=\sigma'$ and there exists $\adef{C}$ such that $A,A' \in \adef{C}$.
We extend it also to labels: $\sigma \simeq_{\mathsf a} \sigma$
and $\sigma|m(\wt{U},A) \simeq_{\mathsf a} \sigma|m(\wt{U}',A')$ iff $m(\wt{U}) \simeq_{\mathsf a} m(\wt{U}')$
and there exists $\adef{C}$ such that $A,A' \in \adef{C}$.
%except that $\eqdot_\alpha$ is used instead of $\eqdot$.
\end{itemize}

\ifconf
\else

\noindent The following Proposition formalizes the correspondence 
between $\lred{\alpha}$ and $\absred{\alpha}$: 
\emph{1.} all $\lred{\alpha}$ transitions are present also in $\absred{\alpha}$
(up-to application of the abstraction 
function $\abst{}$ to configurations) and
\emph{2.} all the abstract computations $\State_0 \absred{\alpha_1} \ldots 
\absred{\alpha_n} \State_n$ have a corresponding concrete computation
$\State'_0 \lred{\alpha'_1} \ldots \lred{\alpha'_m} \State'_m$
in which they can be embedded.
\begin{proposition}
\label{prop:abssemantics}
%Let $\State$ be a configuration such that for every
%actor name $A$ occurring in $\State$, there exists also 
%an actor instance $A \triangleright (P, \varnothing, q)$ 
%in $\State$. The two following statements hold:
Let $\State$ and $\State_0$ be a decorated configuration
and an initial 
decorated configuration $\aleph \; \triangleright \; (0:P, \varnothing, \varepsilon)$
of {\actsl}, respectively.
\begin{enumerate}
\item
If $\State \lred{\alpha} \State'$ then 
$\abst{\State} \absred{\alpha} \abst{\State'}$;
\item
if $\State_0 \absred{\alpha_1} \ldots 
\absred{\alpha_n} \State_n$ then there exists a %corresponding concrete 
computation
$\State'_0 \lred{\alpha'_1} \State'_1 \ldots \lred{\alpha'_m} \State'_m$ and an injection $I$ 
such that, for all $1\leq i \leq n$, 
we have $\alpha_i \simeq_{\mathsf a} \alpha'_{I(i)}$ and 
$\proj{\State_i}{\alpha_i} \simeq_{\mathsf a} \proj{\State'_{I(i)}}{\alpha'_{I(i)}}$.
%$\abst{\State} \abtrans{\mulset}{\mulset'} A \triangleright (P, \varnothing, q),\StateT$ then
%$\State \trans{\mulset'}{} A' \triangleright (P', \varnothing, q'),\State'
%\trans{\mulset''}{} \State''$
%for some
%there exist 
%$P'$, $\State'$, $\State''$, $\mulset'$, $\mulset''$ and $\mulset'''$
%such that
%\begin{enumerate}
%\item
%$\State \trans{\mulset'}{} A' \triangleright (P', \varnothing, q'),\State'
%\trans{\mulset''}{} \State''$
%with 
%$P' \simeq P$ and
%\item
%$\State \trans{\mulset''}{} \State''$ and there exists  
%$\mulset'''$ such that 
%$\mulset \simeq \mulset'''$ for some $\mulset'''$ s.t.
%$\mulset''' \subseteq \mulset' \uplus \mulset''$.
%\end{enumerate}
\end{enumerate}
\end{proposition}

\begin{proof}
The first item trivially holds because the new
rules used in the definition of $\absred{\alpha}$
are (strictly) more general than the corresponding
rules used in the definition of $\lred{\alpha}$.

%{\bf [THIS PART OF THE PROOF MUST BE REWRITTEN
%AS WE HAVE CHANGED THE DEFINITION OF THE TRANSITION
%SYSTEM LABELED WITH MULTISETS]}
The second item is proved by induction on the length of the computation 
$\State_0 \absred{\alpha_1} \ldots 
\absred{\alpha_n} \State_n$.

If $n=1$ then $\State_0 \absred{\alpha_1} \State_1$
with $\alpha_1=1$ or $\alpha_1=1|m(\wt{U},A)$ 
and $\proj{\State_1}{\alpha}=P'$, where $P'$
is an immediate derivative of the main process $P$.
It is trivial to see that the same transition
is present in the concrete decorated semantics:
namely, $\aleph \; \triangleright \; (0:P, \varnothing, \varepsilon) \lred{\alpha'_1} \State'_1$ with 
$\alpha_1 \simeq_{\mathsf a} \alpha'_1$ and 
$\proj{\State'_1}{\alpha'_1} \simeq_{\mathsf a} P'$.

If $n>1$
%In the inductive case we assume that the second item holds for $n>1$.
we consider $\State_0 \absred{\alpha_1} \ldots 
\absred{\alpha_n} \State_n\absred{\alpha_{n+1}} \State_{n+1}$. 
The inductive hypothesis guarantees the existence of the concrete computation 
$\State'_0 \lred{\alpha'_1} \ldots \lred{\alpha'_m} \State'_m$ and of the injection $I$ 
such that, for all $1\leq i \leq n$, 
we have $\alpha_i \simeq_{\mathsf a} \alpha'_{I(i)}$ and 
$\proj{\State_i}{\alpha_i} \simeq_{\mathsf a} \proj{\State'_{I(i)}}{\alpha'_{I(i)}}$.
We now proceed by case analysis on the last number of the sequence
in $\alpha_{n+1}$.

If the number is 0, then $\alpha_{n+1}=\sigma\cdot 0$ and the transition is obtained
by applying rule \rulename{inst$_a$}
on a message $m(\wt{U},\sigma,A')$.
The presence of this message in one of the queues in $\State_n$
guarantees the existence of $1\leq j\leq n$
such that $\alpha_j=\sigma|m(\wt{U},A')$.
In the concrete computation we have $\alpha'_{I(j)}\simeq_{\mathsf a} \sigma|m(\wt{U},A')$.
This means that the same message (up-to $\simeq_{\mathsf a}$) is in the 
queue of an actor $A''$ such that $A',A'' \in {\adef{C}}$ in the concrete state $\State'_{I(j)}$.
We have two subcases: either such method invocation
is instantiated by the actor $A''$ during the concrete
computation $\State'_0 \lred{\alpha'_1} \ldots \lred{\alpha'_m} \State'_m$
or not. In the first case, there exists $I(j) < l \leq m$ such that 
$\alpha'_l=\sigma\cdot 0$ that instantiates
the method. The thesis is proved simply by extending the injection
$I$ with $I(n+1)=l$ and observing that $\proj{\State_{n+1}}{\sigma\cdot 0}  \simeq_{\mathsf a} \proj{\State'_l}{\sigma\cdot 0}$.
If the method invocation was not instantiated, it is in the
queue of the actor $A''$ in the configuration $\State'_m$. It is sufficient
to apply the same reasoning on an extension of the concrete computation
that consumes the messages in front of the method invocation with sequence $\sigma$
and that finally instantiates it.
Such extension exists because processes are finite and non-blocking.

If the number is not 0, we discuss only the case in which $\alpha_{n+1}=\sigma\cdot k$
(with $k > 0$)
because the case $\alpha=\sigma\cdot k|m({\wt U},A)$ is treated similarly.
In the computation $\State_0 \absred{\alpha_1} \ldots 
\absred{\alpha_n} \State_n$ is guaranteed the presence of a label
containing $\sigma\cdot k-1$, i.e.
there exists $1\leq j\leq n$ such that 
the label $\alpha_j$ contains $\sigma\cdot k-1$.
The process $\proj{\State_j}{\alpha_j}$ is the process
that has just performed the action labeled with $\sigma\cdot k-1$
and that performs the action in the transition $\State_n\absred{\alpha_{n+1}} \State_{n+1}$
because $\alpha_{n+1} = \sigma\cdot k$.
By inductive hypothesis $\proj{\State_j}{\alpha_j} \simeq_{\mathsf a} \proj{\State'_{I(j)}}{\alpha'_{I(j)}}$
hence a process ready to execute an action labeled with the sequence $\sigma\cdot k$
occurs also in the concrete state $\State'_{I(j)}$.
We now consider two subcases.
\begin{itemize}
\item
There exists no label $\alpha'_l$
containing $\sigma\cdot k$.
In this case 
the process $\proj{\State'_{I(j)}}{\alpha'_{I(j)}}$
still occurs in $\State'_m$.
Hence it is possible to extend the computation 
$\State'_0 \lred{\alpha'_1} \ldots \lred{\alpha'_m} \State'_m$
with $\State'_m \lred{\sigma\cdot k} \State'_{m+1}$
in such a way that $\proj{\State'_{m+1}}{\sigma\cdot k}\simeq_{\mathsf a}\proj{\State_{n+1}}{\sigma\cdot k}$.
The thesis is proved simply by extending the injection
$I$ with $I(n+1)=m+1$.
\item
There exists $I(j) < l \leq m$
such that $\alpha'_l$ contains $\sigma\cdot k$.
In this case it is not guaranteed
that $\proj{\State'_{l}}{\alpha'_l}\simeq_{\mathsf a}\proj{\State_{n+1}}{\sigma\cdot k}$, due to
nondeterminism. For this reason we construct from 
$\State'_0 \lred{\alpha'_1} \ldots \lred{\alpha'_m} \State'_m$
another concrete computation
%$\State''_0 \lred{\alpha''_1} \ldots \lred{\alpha''_t} \State''_t$
that satisfies our thesis.
The first transformation that we apply to $\State'_0 \lred{\alpha'_1} \ldots \lred{\alpha'_m} \State'_m$
consists of the cancellation
of the transition $\alpha'_l$ and of all the other transitions
that depends on it. Namely, there are two kinds of transitions that 
must be removed: (i) those labeled with a sequence having a prefix $\sigma\cdot r$ such that
$r \geq k$ and (ii) those causally dependent on the instantiation of messages
that are in $\State'_{l}$ inside the queue of the actor containing the process
decorated with $\sigma\cdot k$.
Let
$\State'_0 \lred{\alpha'_1} \ldots \lred{\alpha'_{l-1}}
\State'_{l-1}  \lred{\alpha''_l} \State''_l \ldots \lred{\alpha''_s} \State''_s$
be the concrete computation obtained after this elimination of transitions.
We now extend such computation 
by letting the process labeled with $\sigma\cdot k-1$
to execute the expected action labeled with $\sigma\cdot k$.
Namely, we add
the transition $\State''_s \lred{\sigma\cdot k} \State''_{s+1}$
such that $\proj{\State''_{s+1}}{\sigma\cdot k} \simeq_{\mathsf a} \proj{\State_{n+1}}{\sigma\cdot k}$.
Then, we extend the computation by performing at least
all the transitions removed for the reason (ii) above.
This extension exists because all processes are finite and
nonblocking and because the considered transitions causally depend on
messages that are in $\State''_{s+1}$ inside the queue of the actor
containing the process decorated with $\sigma\cdot k$.
Let
$\State'_0 \lred{\alpha'_1} \ldots \lred{\alpha'_{l-1}}
\State'_{l-1}  \lred{\alpha''_l} \State''_l \ldots \lred{\alpha''_s} \State''_s
\lred{\alpha''_{s+1}}\ldots \lred{\alpha''_{t}}\State''_t$
be the obtained computation.
The thesis is proved by considering this last concrete
computation, a rearrangement of the injection $I$
that maps to their new positions the transitions in its codomain
that belong to the group (ii), and by extending it with $I(n+1)=l$.\qedhere
\end{itemize}
\end{proof}

\noindent As a direct consequence we have that the abstract semantics preserves
both termination and 
control-state reachability.
%\begin{corollary}
%Given an actor program of {\actsl} we have that:
%\begin{itemize}
%\item
%it has an infinite computation if and only it has an infinite computation according
%to the abstract semantics;
%\item
%it can reach the program $P$ if and only if there exists a renaming of 
%
%it can reach the program $P$ 
%
%and $\State_0$ be a decorated configuration
%and an initial 
%decorated configuration $\aleph \; \triangleright \; (0:P, \varnothing, \varepsilon)$
%of {\actsl}, respectively.
%\end{corollary}
\fi
%\begin{theorem}
%\label{theo:asbsemantics}
%Let $\State$ be a state of a transition system 
%of a program in {\actsl}.
%%The two following statements hold:
%\begin{itemize}
%\item
%$\State$ \emph{terminates} in the concrete
%transition system if and only if $\abst{\State}$ 
%\emph{terminates} in the abstract transition system;
%\item
%given a process $P$, 
%there exist $A'$, $q'$, and $\State'$
%such that
%$\State \lred{}^* A' \triangleright (P, \varnothing, q'),\State'$
%if and only if
%there exist $A''$, $q''$, and $\State''$
%such that
%$\abst{\State} \longrightarrow_\alpha^* A'' \triangleright (P, \varnothing, q''),\State''$.
%\end{itemize}
%\end{proposition}
%
\ifconf
%\else
%\begin{proof}
%The second item is a direct consequence of Proposition~\ref{prop:abssemantics}.
%For the first item we first observe that if $\State$ has an infinite
%computation then the same also holds for $\abst{\State}$
%(see the first item of 
%Proposition~\ref{prop:abssemantics}).
%We now assume that $\State$ terminates.
%This implies the existence of an upper bound $k$ such that
%if $\State \trans{\mulset}{}  \State'$ then $|\mulset| \leq k$.
%We now proceed by contradiction assuming that 
%$\abst{\State}$ does not terminate. If this is the
%case, there exists $\StateT$ such that 
%$\abst{\State} \trans{\mulset'}{}  \StateT$ with $|\mulset'| > k$.
%By the second item of Proposition~\ref{prop:abssemantics}
%we have that there exists $\State''$ such that $\State \trans{\mulset''}{}  \State''$ such that $|\mulset''| \geq |\mulset'| > k$.
%But this contradicts the assumption on the upper bound $k$.
%\qed
%\end{proof}
\fi
It remains to prove that termination and process reachability is
decidable for the abstract semantics. To this aim, we consider
a transition system $\absred{}$ obtained by removing the 
labels from the transitions $\absred{\alpha}$.
On this transition system we 
define $\preceq_{\mathsf a}$ as
a variant of the ordering
$\preceq$ defined in the previous section
in such a way that $({\cal S},\absred{},\preceq_{\mathsf a})$
turns out to be a well-structured transition system (for configurations 
of stateless programs). 
%We first redefine the notions of 
%Section~\ref{sec.decidability}.
%,
%for every abstract transition system $({\cal S},\absred{})$
%of a stateless program. 
Let:
\begin{itemize}
\item[--]
Let $\leq_{\mathsf a}$ be the following relation on message queues:
%$\leq_\alpha$ be the %least 
%relation %on queues 
%defined as $\leq$ 
%in Section~\ref{sec.decidability}, with %by using 
%$\simeq_\alpha$ instead of $\simeq$.

{\small
\[
\bigfract{
	{\it there \; exist} \; i_1 < i_2 < \cdots < i_k \leq h \; 
	{\it s.t. \; for} \; j \in 1..k,  \; \; m_j(\wt{U_j},\sigma_j,A_j) 
	\simeq_{\mathsf a} n_{i_j}(\wt{V}_{i_j},\sigma'_{i_j},A'_{i_j})
	}{
	m_1(\wt{U_1},\sigma_1,A_1) \ldots m_k(\wt{U_k},\sigma_k,A_k) \leq_{\mathsf a} 
	n_1(\wt{V_1},\sigma'_1,A'_1) \ldots n_h(\wt{V_h},\sigma'_h,A'_h)
	}
\]
}

\item[--]
Let $\preceq_{\mathsf a}$ be the ordering:
%\[
%\bigfract{
%\exists \; i_1 < i_2 < \cdots < i_k \leq h \; 
%	{\it\ s.t.\ } \; \forall j \in 1..k,  \; \; (A_j \simeq_\alpha A_{i_j}'\ \wedge\ P_j \simeq_\alpha P_{i_j}'
%	\ \wedge\ 
%	q_j \leq_\alpha q_{i_j}')
%	}{
%	 A_1 \triangleright (P_1, \varnothing, q_1), \cdots , A_k \triangleright (P_k, \varnothing, q_k) 
% \; \preceq_\alpha \; A'_1 \triangleright (P_1', \varnothing, q_1'), \cdots , A'_h \triangleright 
% (P_h', \varnothing, q_h')
% }
%\]
\[
\bigfract{A_i,A'_{j_i} \in {\adef{C}_i} \ \ P_i \simeq_{\mathsf a} P_{j_i}'
%\ \ \sigma_i = \sigma_{j_i}'
\ \  \mbox{\rm and}  
	\ \ q_i \leq_{\mathsf a} q_{j_i}' \ \ \mbox{\rm for } {i \in 1..\ell}, \; 
	1 \leq j_1 < j_2 < \cdots < j_\ell \leq \kappa
	}{
	 A_1 \triangleright (\sigma_1:P_1, \varnothing, q_1), \cdots , A_\ell \triangleright (\sigma_\ell:P_\ell, \varnothing, q_\ell) 
 \; \preceq_{\mathsf a} \; A'_1 \triangleright (\sigma'_1:P_1', \varnothing, q_1'), \cdots , A'_\kappa \triangleright 
 (\sigma'_\kappa:P_\kappa', \varnothing, q_\kappa')
 }
\]

\end{itemize}

Next, we observe that Lemma~\ref{prop.finteterms}
can be adapted to the case of unbounded actors
by using $\simeq_{\mathsf a}$ instead of $\simeq$.
Namely, let $T$ be either a process or a method invocation $m(U_1, \cdots , U_n, \sigma, A)$
of a stateless program and
let ${\cal T} = \{ T \rho_1, T\rho_2, T\rho_3, \cdots \}$ be
such that $i \neq j$ implies $T \rho_i \not \simeq_{\mathsf a} T \rho_j$. 
Proceeding as in the proof of Lemma~\ref{prop.finteterms},
%by adding the observation that a program has a finite set
%of actor definitions, 
we prove that
${\cal T}$ is finite.

\begin{theorem}
\label{thm.decidableforunbounded}
Given a stateless program ${\cal S}$
%Let $({\cal S},\absred{})$ be the abstract transition system of a 
%program in {\actsl}.
%Then 
we have that $({\cal S}, \absred{}, \preceq_{\mathsf a})$ is a well-structured transition system.
\end{theorem}

\ifconf
\else

\begin{proof}
The proof is as in Theorem~\ref{thm.decidablestatelessandfinite}
with few differences that are discussed below.

In part {\bf (1)} the unique difference is in the 
last part where the coordinatewise order $\sqsubseteq^\ell$
on sequences (of length $\ell$) of queues of terms 
is used. As we now consider configurations with an unbounded 
number of actors, instead of configurations with a bounded
number $\ell$ of actors, we need to resort to the 
embedding $\sqsubseteq_{\mathsf a}$ defined as follows:
\[
\bigfract{
	q_i \leq_{\mathsf a} q_{j_i}' \quad \mbox{for } \; i \in 1..\ell, \; 
	1 \leq j_1 < j_2 < \cdots < j_\ell \leq \kappa
	}{
	 (q_1, \cdots , q_\ell) 
 \; \sqsubseteq_{\mathsf a} (q_1', \cdots , q_\kappa')
 }
\]
The final contradiction of part {\bf (1)}
is now reached by observing that by Highman's lemma, 
also $\sqsubseteq_{\mathsf a}$ is a well-quasi-ordering, as a consequence
of the well-quasi-ordering $\leq_{\mathsf a}$.

In part {\bf (2)} the unique difference
is for the monotonicity 
transitions due to rules \rulename{invk-a} and \rulename{invk-sa}. The greater configuration
is guaranteed to have a program ready to perform
a corresponding method invocation, but this could 
be addressed to a different actor.
In fact, the ordering $\preceq_{\mathsf a}$
does not preserve actor names as it was for $\preceq$ 
in the proof of Theorem~\ref{thm.decidablestatelessandfinite}. 
But $\preceq_{\mathsf a}$ preserves at least actor
classes. As the abstract transition
system $\absred{}$ allows a term $m(\wt{U},\sigma,A)$ to be
introduced in the queue of any of the actor 
belonging to the same class of $A$,
the method invocation executed by the greater
configuration can be introduced in the 
queue of the actor corresponding to the target
of the method invocation executed by the
smaller configuration.
\end{proof}
\fi
In the light of the results on well-structured transition
systems recalled at the beginning of Section~\ref{sec.decidability},
this theorem proves the decidability of termination for the abstract semantics.
To prove the decidability of process reachability
we need to prove that a finite basis for predecessors
is effectively computable.
%We need to
%consider a more sophisticated 
%algorithm for computing the 
%predecessors of a configuration.
%
%In the light of Theorem~\ref{thm.decidableforunbounded}, it is possible to 
%decide the termination for the
%abstract transition system of a stateless
%program. As termination is preserved
%by the abstract semantics (see Proposition~\ref{prop:asbsemantics})
%we can conclude that termination is also decidable
%for the concrete transition system of a stateless program.
%
%\medskip
%
%We complete this section by demonstrating the decidability of
%control-state reachability for the
%well-structured transition system 
%$({\cal S}, \absred{}, \preceq_\alpha)$
%of a stateless program
%(see the definition after Lemma~\ref{lem.pred}).
%The proof is similar to the one of 
%Theorem~\ref{thm.cs-reachability}, with the difference that it is needed
%a more sophisticated algorithm for computing the 
%predecessors of a configuration.

\begin{lemma}
\label{lem.predunbounded}
Let $({\cal S}, \absred{}, \preceq_{\mathsf a})$ be a 
well-structured transition system 
of a program in {\actsl}, and let $\State \in
{\cal S}$. Then there is a finite set ${\cal X}$
such that, for every $\State' \succeq_{\mathsf a} \State$ and 
$\State'' \in \pred{\State'}$, there is $\StateT \in {\cal X}$
with $\StateT \preceq_{\mathsf a} \State''$. ${\cal X}$ can be effectively computed.
\end{lemma}

\ifconf
\else

\begin{proof}
The computation of ${\cal X}$ must extend the 
construction presented in the proof of Lemma~\ref{lem.pred}
in two ways.

The first extension derives from the fact that, 
differently from the ordering $\preceq$ considered 
in Lemma~\ref{lem.pred}, if $\State \preceq_{\mathsf a}  \State'$ 
it could be possible for $\State'$ to have strictly
more actors than $\State$.
In these case, it is possible that the predecessor
differs from its successor $\State'$ for actors
that are not present in $\State$.
We can cope with this problem
by applying the procedure described in the proof
of Lemma~\ref{lem.pred} not only to the configuration
$\State$, but to all the configurations that can be obtained
by extending $\State$ with
one or two additional actors belonging to one of the
finite classes of the considered program. 
In fact, at most two actors are modified by one transition.
%in a computable set 
%representing a finite basis for all the configurations 
%$\State' \succeq_{\mathsf a} \State$ with strictly more actors
%than $\State$.
%This set includes all the extensions of $\State$ with
%one or two additional actors belonging to one of the
%finite classes of the considered program. 
%%There
%%finitely many combination of actor names up-to 
%%$\preceq_\alpha$.
Each of these additional actors executes a process obtained
by applying a renaming $\rho$ to a suffix of 
one of the method definitions
of the corresponding class. As observed
above, there are finitely many processes that can 
be obtained up-to $\simeq_{\mathsf a}$.
Finally, the additional actors
have a queue including at most one
method invocation (in fact, at most one message
can be consumed in one transition). Also in this
case, by considering the 
method definitions of the actor class,
it is easy to see that there are finitely
many different method invocations up-to $\simeq_{\mathsf a}$.

The second extension %{\bf(1)} 
is trivial and deals with the fact
that in the abstract semantics a method invocations can 
be introduced in the queue of any of the actors
belonging to the same class of the expected target actor.
So the procedure of the proof of Lemma~\ref{lem.pred} for computing
%The set 
${\cal X}$ %is computed by applying the 
must be extended to consider also this kind
of transitions.
%procedure of the proof of Lemma~\ref{lem.pred}
%with the extension {\bf (1)}, to $\State$
%as well as to all the other configurations  
%computed as in the extension {\bf (2)}.
\end{proof}
\fi

\ifconf
%Let the \emph{control-state reachability problem} be to decide, given two states ${\State}$ and ${\StateT}$ of a well-structured transition system with well-quasi ordering
%$\preceq$,
%whether there is $\StateT' \succeq \StateT$ such that $\State \lred{}^* \StateT'$.
It turns out that 
control-state reachability
is decidable for the
abstract transition system of {\actsl}.
\else
From Theorem~\ref{thm.decidableforunbounded}, this last Lemma 
and the results on well-structured transition
systems recalled at the beginning of Section~\ref{sec.decidability},
we can conclude that also control-state reachability,
besides termination as already commented above,
is decidable for the abstract semantics. From Proposition~\ref{prop:abssemantics}
we have already concluded that the abstract semantics preserves
termination and control-state reachability w.r.t. the concrete 
semantics. Hence, we have that termination and 
control-state reachability are decidable for
% the abstract transition system of a 
stateless actor program. 
\fi

The decidability of control-state reachability entails the decidability of
process reachability.
In fact,
given a process $P$, 
the reachability of a configuration
$A \triangleright (P', \varphi, q),\State$
with $P'$ equal to $P$ up-to renaming
 of variables and actor names
can be solved in the abstract transition
system simply by checking the control-state
reachability of at least one 
of the following states.
Let ${\adef{C}_1},\ldots,{\adef{C}_n}$
be the actor classes of the considered actor system
and let 
$A_1,\cdots,A_n$ 
be such that $A_i \in {\adef{C}_i}$.
We consider the following finite set of states:
\[
\begin{array}{lll}
{\cal S} & = & \{\ A_i \triangleright (Q_i, \varnothing, \varepsilon)\ \mid \
1 \leq i \leq n,
%\\
%& & \qquad
\quad
\mbox{$Q_i$ is a suffix of a method definition} \\ 
& & \qquad
 \mbox {in the class
${\adef{C}_i}$
%}
%\\
%& & \qquad\qquad
%\mbox{
and it is
equal to $P$ up-to renaming}
\ \}
\end{array}
\]
%From the decidability of the termination
%and of the process reachability
%problems
%for the abstract transition system
%we can conclude their decidability
%for the concrete semantics.
%By Proposition~\ref{prop:asbsemantics},
%this problem is preserved by the abstract
%semantics.
\ifcamera
Note that 
control-state reachability 
is not preserved by the abstract semantics.
In fact, the abstract transition system is guaranteed to 
execute the same method invocations, but this can 
be done in a different order and also by different 
actors.
\else

%We complete this section by observing that
%the other properties considered in Section~\ref{sec.decidability}
%(control-state maintainability,
%inevitability, and
%control-state reachability) 
%are not preserved by the abstract semantics.
%In fact, the abstract transition system is guaranteed to 
%execute the same method invocations, but this can 
%be done in a different order and also by different 
%actors.
%So, for instance, if a configuration is reachable
%according to the abstract semantics, or a computation
%traverses only configurations covering a given 
%set of states, we cannot conclude
%that the same holds in
%the concrete transition system.
\fi

%\begin{lemma}
%\label{lem.pred}
%Let $({\cal S}, \lred{}, \preceq)$ be a well-structured transition system 
%of a stateless program with a bounded number of actors, and let $\State \in
%{\cal S}$. Then there is a finite set ${\cal X} \subseteq \pred{\State}$ 
%such that, for every $\State' \in \pred{\State}$, there is $\StateT \in {\cal X}$
%with $\StateT \preceq \State'$. ${\cal X}$ can be effectively computed.
%\end{lemma}

%Let $({\cal S}, \absred{}, \preceq_\alpha)$ be a 
%well-structured transition system 
%of a stateless program, and let $\State \in
%{\cal S}$. 
%By extending the proof of Lemma~\ref{lem.pred}
%we show that a finite set ${\cal X} \subseteq \pred{\State}$ 
%such that, for every $\State' \in \pred{\State}$, there is $\StateT \in {\cal X}$
%with $\StateT \preceq_\alpha \State'$. ${\cal X}$ can be effectively computed.
%

%\ifcamera
%\else 
\section{Conclusions} %and future work}
\label{sec.conclusions}
To the best of our knowledge this paper contains a first systematic study on the computational power of Actor-based languages.
We have focussed on the  pure asynchronous FIFO queueing and dequeuing of method calls between actors in the context of a nominal calculus 
which features the dynamic creation of variable names that can be passed around.
The results proved in this paper can be summarized as follows:
\begin{itemize}
\item
we identified two small but Turing powerful fragments of our Actor language:
the fragment in which only boundedly many actors can be created,
and the fragment in which fields cannot be updated;
\item
we have proved that the fragment obtained as intersection of the 
two above sublanguages is not Turing complete, as properties like
termination and control-state reachability turn out to be decidable;
\item
if Actors are stateless, the language has decidable termination
and control-state reachability even if we consider unboundedly many
Actors. 
\end{itemize}

\noindent We conclude by mentioning 
relevant lines
for future research. Recent work have identified more expressive Actor interaction
mechanisms based an asynchronous method calls implemented by means of the so-called
future variables \cite{BoerCJ07}: we plan to investigate the impact of this
Actor-based synchronization mechanism on our (un)decidability results.
We also plan to extend our study of expressiveness to 
primitives like the release statements in \cite{Johnsen07}. 
These statements support the so-called cooperative shceduling
of method invocations:
method executions can release the 
control of the Actor in such a way that other method executions
can be instantiated or resumed.

%supporting the so-called
%cooperative scheduling of method executions.
%This can be obtained, for instance,
%with the release statements in \cite{Johnsen07}. 
%These statements are used by method executions to release the 
%control of the Actor, in such a way that other method executions
%can be instantiated or resumed.
%relasing 
%
%Future work concerns a study of  Actor-based synchronization mechanisms like
%future variables \cite{BoerCJ07} and release statements \cite{Johnsen07}.

\bibliographystyle{abbrv}
\bibliography{Mmajid}

\end{document}